\pdfoutput=1

\documentclass[10pt,conference]{IEEEtran}
\newcommand{\CLASSINPUTbottomtextmargin}{1.0in}
\usepackage[T1]{fontenc}
\usepackage{authblk}
\usepackage[latin1]{inputenc}
\usepackage[all]{xy}
\usepackage{mathtools}
\usepackage{amsmath,amssymb}
\usepackage{bbm}
\usepackage{graphicx}
\usepackage{epstopdf}
\usepackage{color}
\usepackage{cite}
\usepackage{multirow,tabularx}
\usepackage{supertabular}
\usepackage{mdframed}
\usepackage{ifthen}
\usepackage{times}
\usepackage{stackrel}
\usepackage{subcaption}

\usepackage{algorithm}
\usepackage{algorithmicx}
\newtheorem{theorem}{{\bf Theorem}}
\newtheorem{lemma}{{\bf Lemma}}
\newtheorem{proposition}[theorem]{Proposition}

\newenvironment{proof}[1][Proof]{\begin{trivlist}
\item[\hskip \labelsep {\bfseries #1}]}{\end{trivlist}}

\newtheorem{defn}{Definition}

\newenvironment{remark}[1][Remark]{\begin{trivlist}
\item[\hskip \labelsep {\bfseries #1}]}{\end{trivlist}}

\newcommand{\qed}{\nobreak \ifvmode \relax \else
      \ifdim\lastskip<1.5em \hskip-\lastskip
      \hskip1.5em plus0em minus0.5em \fi \nobreak
      \vrule height0.75em width0.5em depth0.25em\fi}

\newcommand{\barr}{\begin{array}}
\newcommand{\earr}{\end{array}}

\newcommand{\benum}{\begin{enumerate}}
\newcommand{\eenum}{\end{enumerate}}

\newcommand{\bit}{\begin{itemize}}
\newcommand{\eit}{\end{itemize}}

\newcommand{\bc}{\begin{center}}
\newcommand{\ec}{\end{center}}

\newcommand{\bdes}{\begin{description}}
\newcommand{\edes}{\end{description}}

\newcommand{\efig}{\end{figure}}

\newcommand{\bemq}{\begin{quote} \begin{em}}
\newcommand{\eemq}{\end{em} \end{quote}}

\newcommand{\bmp}{\begin{minipage}}
\newcommand{\emp}{\end{minipage}}

\newcommand{\eqn}[1]{(\ref{#1})}

\newcommand{\fgr}[1]{Fig.~\ref{#1}}
\newcommand{\tbl}[1]{Table~\ref{#1}}

\newcommand{\brac}[1]{\left({#1}\right)}
\newcommand{\sbrac}[1]{\left[{#1}\right]}
\newcommand{\cbrac}[1]{\left\{{#1}\right\}}

\newcommand{\tendsto}{\rightarrow}
\newcommand{\belongs}{\in}
\newcommand{\expect}[1]{{\mathbb{E}}\!\left[{#1}\right]}

\newcommand{\prob}[1]{\text{Pr}\brac{#1}}

\DeclareMathOperator*{\argmax}{arg\,max}
\newcommand{\abs}[1]{\left\lvert{#1}\right\lvert} % Usage: \abs{x}

\newcommand{\Given}{\Big\arrowvert}
\newcommand{\strategy}{X}
\newcommand{\utility}{U}

\newcommand{\noise}{P_0}

\IEEEoverridecommandlockouts
\begin{document}

%*************************************************************
%*****    FRONT MATTER
%*************************************************************

%----------------------------------------------------------------------
%%% TITLE & AUTHORS
%----------------------------------------------------------------------
%\title{Learning Optimal Channel Assignment in D2D Wireless Networks Using Noisy Potential Games}
\title{Optimal Distributed Channel Assignment in D2D Networks Using Learning in Noisy Potential Games}
%\title{\msa{Optimal Distributed Learning Algorithm for Stochastic Optimization}}
%
\author{Mohd. Shabbir Ali, Pierre Coucheney, and Marceau Coupechoux
\thanks{Mohd. Shabbir Ali and M. Coupechoux are with LTCI, Telecom ParisTech, Universit\'e Paris-Saclay.
Emails: mdshabbirali88@gmail.com, marceau.coupechoux@telecom-paristech.fr. 
P. Coucheney is with UVSQ, David-Lab, France. 
Email: pierre.coucheney@uvsq.fr. This work was supported by NetLearn ANR project (ANR-13-INFR-004) and the Indo-French CEFIPRA project "D2D for LTE-Advanced"}}
%

%}

\maketitle

%----------------------------------------------------------------------
%%% ABSTRACT
%----------------------------------------------------------------------
\begin{abstract}
We present a novel solution for Channel Assignment Problem (CAP) in Device-to-Device (D2D) wireless networks that takes into account the throughput estimation noise. CAP is known to be NP-hard in the literature and there is no practical optimal learning algorithm that takes into account the estimation noise. In this paper, we first formulate the CAP as a stochastic optimization problem to maximize the expected sum data rate. To capture the estimation noise, CAP is modeled as a noisy potential game, a novel notion we introduce in this paper. Then, we propose a distributed Binary Log-linear Learning Algorithm (BLLA) that converges to the optimal channel assignments. Convergence of BLLA is proved for bounded and unbounded noise. Proofs for fixed and decreasing temperature parameter of BLLA are provided. A sufficient number of estimation samples is given that guarantees the convergence to the optimal state. We assess the performance of BLLA by extensive simulations, which show that the sum data rate increases with the number of channels and users. Contrary to the better response algorithm, the proposed algorithm achieves the optimal channel assignments distributively even in presence of estimation noise.

\end{abstract}

%*************************************************************
%*****    BODY TEXT
%*************************************************************
\section{Introduction}
Ever increasing demand for higher data rates of mobile users and scarcity of wireless frequency spectrum 
is making efficient utilization spectrum resources increasingly critical.
Device-to-Device (D2D) networks increase the utilization of the spectrum resources by providing spatial spectrum reuse~\cite{asadi2014survey,tehrani2014device,song2014game,ye2015distributed,li2014coalitional,chen2014coalition,cai2014distributed,huang2014resource,xu2013efficiency,wang2016optimal,maghsudi2016hybrid,gao2014joint}.
In a D2D network, D2D users reuse the radio resources allocated to traditional cellular Users (UEs).
The cellular UEs communicate with the Base Station (BS) while the D2D UEs communicate among themselves without or with limited help from BS.
 This is possible provided that the interference caused by the D2D UEs to cellular UEs is limited. 
A crucial problem in underlay D2D networks is thus to assign channels to UEs to increase its utilization while maintaining low interference.

Channel Assignment Problem (CAP) in D2D networks is challenging due to the lack of  Channel State Information (CSI) of D2D links at the BS and because feedback overheads should be kept at a reasonable level. 
CSI estimation errors that are due to several factors such as randomly varying channel gain, feedback errors, feedback delay errors, and quantization errors~\cite{goldsmith2005wireless} affect the performance of D2D network.  
It is, therefore, essential to have low feedback distributed solutions achieving optimal channel assignment while taking into account CSI or throughput estimation errors.

\subsection{Contributions}
%We provide a summary of our main contributions below.
\begin{itemize}
\item \emph{Novel Approach:} 
Our approach is to learn the optimal channel assignments in a D2D wireless network using a noisy potential game that takes into account the estimation noise. Distributed learning in a noisy environment for CAP is novel.
We consider a Stochastic Optimization Problem (SOP) with the objective to maximize the expected sum data rate of an underlay D2D network. We translate this problem into a noisy potential game. The notion of the noisy potential game is introduced to account for the fact that only noisy estimates of the utility are available to the players. 

\item  \emph{Learning algorithm:} We propose a distributed Binary Log-linear Learning Algorithm (BLLA) for a SOP. 
BLLA solves CAP to achieve an optimal channel assignment, which corresponds to the optimal Nash equilibrium of the game.
The convergence of BLLA is proved for fixed temperature and decreasing temperature parameter. 
We provide a sufficient number of estimation samples that guarantees the convergence for both the cases of bounded noise and unbounded noise. 
Note that for SOPs, BLLA is distributed and more practical when compared to Stochastic Approximation (SA)~\cite{robbins1951stochastic}, Finite Difference SA~\cite{kiefer1952stochastic}, and Simultaneous Perturbation SA~\cite{spall1992multivariate} algorithms, which are centralized and may not be desirable in large networks.
Note that compared to BLLA the algorithm in~\cite{leslie2011equilibrium} considers only the fixed temperature.

\item \emph{Simulations results:} Extensive simulations show that BLLA achieves the maximum sum data rate of the network.
It shows that BLLA tracks well the increase of sum data rate with the increase of UEs and with the increase of the number of channels. 
We also show that contrary to better response algorithm, BLLA converges to the optimum even in presence of estimation noise.

\end{itemize}

\subsection{Related literature survey and comparison}
In this subsection, we discuss and compare different approaches for CAP in the literature.
CAP in wireless networks is a standard problem and it is known to be NP-hard~\cite{garey2002computers}. 
Extensive surveys of CAP can be found for underlay D2D networks~\cite{asadi2014survey} and in various contexts in~\cite{ahmed2014channelSurvey, audhya2011survey}.
%, Orthogonal Frequency Division Multiple Access (OFDMA) systems~\cite{ahmed2014channelSurvey}, heterogeneous networks~\cite{audhya2011survey}, and for cognitive radio networks~\cite{ahmed2014channelSurvey}. 

The CAP solution approaches adopted by the state-of-the-art
channel assignment algorithms are 
dynamic programming~\cite{wang2016optimal}, graph-theoretical and heuristic solutions~\cite{maghsudi2016hybrid,gao2014joint}, game theory~\cite{song2014game,ye2015distributed,li2014coalitional,chen2014coalition,cai2014distributed,huang2014resource,xu2013efficiency},
linear programming (LP), non-linear programming (NLP), and Markov Random Field~\cite{ahmed2014channelSurvey}.
Other approaches  for CAP are neural networks~\cite{chan1994neural}, simulated annealing~\cite{duque1993channel}, tabu search,
genetic algorithms~\cite{audhya2011survey}. %We discuss closely related works in D2D in the following.

 In~\cite{wang2016optimal}, the authors jointly optimize the mode selection and
channel assignment in a cellular network with underlay
D2D communications in order to maximize the weighted sum rate.
A dynamic programming (DP) algorithm is proposed but it is exponentially complex. Therefore, a suboptimal greedy algorithm is proposed. In contrast to our approach, this solution relies on explicit closed form expressions of sum data rate for different channel fading scenarios. 
%In contrast, our approach does not need closed form expressions of the sum data rate. 
Our method can be applied to any fading scenario as it is based on users measured throughput.
%just needs the measured throughput of few users on only two channels in any iteration.
 In~\cite{maghsudi2016hybrid}, a suboptimal graph-theoretical heuristic solution for CAP in D2D networks is proposed. 
The weighted signal sum is maximized using maximum-weighted bipartite matching and interference
sum is minimized using minimum-weighted partitioning. This approach is centralized since the BS uses the partial CSI of all the UEs. 
On the contrary, our approach is distributed, optimal, and maximizes the sum throughput instead of signal sum. 
% Our approach is distributed as BS requires information from only a few users operating on two channels but not on all channels in any iteration.
 In~\cite{gao2014joint}, a heuristic algorithm is proposed for joint mode selection, channel allocation and power allocation in a D2D wireless network. Channel estimation is assumed to be perfect.
% The channel gain of the different links measured through pilot signal is considered to be exact without noise.  %Whereas, our approach takes into account the effect of estimation noise of throughput.
  
%\subsection{Related Game Theoretic Approaches for CAP}
Different game-theoretic models such as non-cooperative games, coalition formation games, and auction games are used to study the radio resource allocation issues in D2D networks~\cite{song2014game,xu2013efficiency,ye2015distributed,li2014coalitional,chen2014coalition,cai2014distributed,huang2014resource}. In~\cite{xu2013efficiency}, a game-theoretical reverse
iterative combinatorial auction is proposed as the allocation mechanism.  
%to optimize the system sum rate over the resource sharing of both D2D and cellular modes of the users. In this auction model, the different channels compete for the D2D links. The auction has proved to be cheat proved that means channels allocated to links with the highest bid. 
However, in this auction, the BS needs to have the bid from all the D2D links that may create a huge feedback overhead.
In~\cite{ye2015distributed}, a pricing mechanism is proposed to maximize the
network throughput under QoS constraint.
%The base stations transmit a pricing signal to the D2D users. The price increases with the interference from the D2D links. The D2D devices utilize the pricing signal to update their strategies  to maximize their data rates. 
However, the algorithm proposed is a heuristic algorithm whose performance is only evaluated through simulations. In contrast, BLLA's convergence is proven theoretically and confirmed through simulations. In~\cite{li2014coalitional}, the uplink resource allocation problem for multiple D2D and cellular users
is modeled as a coalition game. %The utility function combines different modes, mutual interferences, and resource sharing policy. 
Convergence to a Nash equilibrium is proved. However, the equilibrium may be sub-optimal and inefficient. In~\cite{chen2014coalition, cai2014distributed} also, the coalition formation algorithms proposed to jointly solve mode selection and spectrum sharing may not converge to an optimal coalition structure. The contract-based game theoretic mechanism proposed in~\cite{huang2014resource} is evaluated through simulations only. 
%In order to improve the energy efficiency of wireless users, the joint mode selection and spectrum sharing as a coalition formation game is modeled and a coalition formation algorithm is proposed to jointly solve the mode selection and spectrum sharing in a D2D system in~\cite{chen2014coalition}. The algorithm converges to a stable coalition. 
%However, a stable coalition may not be optimal. Whereas, BLLA is proven to be optimal.
%In~\cite{cai2014distributed}, the joint mode selection and spectrum sharing problem are modeled as a coalition formation game. A distributed coalition formation algorithm is proposed and it's performance is evaluated through simulations. 
%In~\cite{huang2014resource}, proposed a contract-based game theoretic mechanism to resolve the unknown channel quality problem by eliminating the incentive of users to report untruthfully with designed service contracts. Its performance is evaluated through simulations only.
%Whereas, as said earlier, our algorithm's performance is proved theoretically and confirmed through simulations.

%Compared to solutions in the literature above, BLLA takes into account the estimation noise, it is distributed, it is provably optimal.

This paper is organized as follows. 
The system model and problem formulation are described in Section~\ref{sec:system_model}.
A noisy potential game framework is developed in Section~\ref{sec:NPG}.
BLLA and its convergence results are given in Section~\ref{sec:learning}.
Simulation and conclusions are presented in Sections~\ref{sec:simulations} and~\ref{sec:conclusions}, respectively. 
Proofs are in Appendix of our arXiv paper~\cite{aliCAP2017}.

\section{D2D Cellular Network Model}\label{sec:system_model}
In this section, we describe the D2D cellular network model as shown in~\fgr{fig:layout}.
This figure shows downlink (DL) and uplink (UL) models.
We consider a single base station (BS) and two types of UEs: 
$(i)$ cellular UEs (UECs) that communicate with the BS and 
$(ii)$ D2D UEs (UEDs) that communicate with other UEDs.
%All UEs are randomly located in the region. 
The set of UEs is denoted as $\mathcal{D}$.
We consider a set of orthogonal frequency channel bands $\mathcal{F}$. 
The UECs are assigned different channels by the BS, whereas UEDs reuse these channels.
A UE transmits on a single channel. 
The UEs that transmit on the same channel $c \in \mathcal{F}$ cause interference to each other, the amount of which depends on channel gains between transmitters and receivers.
%The interference scenario is shown in~\fgr{fig:layout} for both DL and UL. 
%In DL, BS transmits a signal to UEC while causing interference at the receivers of UEDs that use the same channel as shown in~\fgr{fig:layout}. 
%A UED transmits a signal to its receiver (Rx) while causing interference to the UECs and receivers of other UEDs that are on the same channel.
%In UL, UEC transmits a signal to the BS while causing interference to UEDs, see~\fgr{fig:layout}.
%Also, UEDs cause interference to the BS and other UEDs in UL.
%The amount of interference caused depends on the channel gain, which we describe below.
%The channel model is described below.

%\graphicspath{{Figures//}}
\begin{figure}
  \centering
\includegraphics[width=0.8\linewidth]{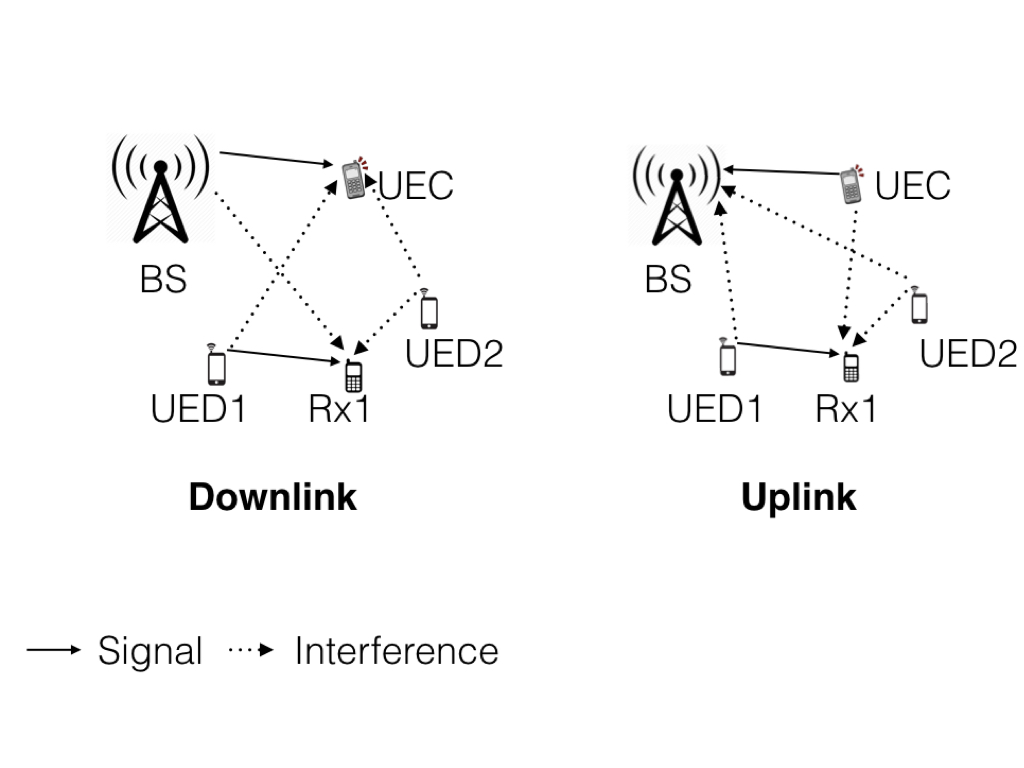}
  \caption{D2D cellular network layout model.}
  \label{fig:layout}
%  \vspace{-1cm}
\end{figure}

\subsection{Channel Model}

We consider a channel model that captures the effect of path-loss, shadowing, and small-scale fading. 
%%The shadow fading component is a constant multiplicative factor as the scenario considered is static without any mobility. 
%Formally, the channel model considered is~\cite{goldsmith2005wireless}:
%\begin{equation}
%g(d)  = \min\cbrac{1,K\abs{d}^{-\eta}e^{y}\hbar},
%\end{equation}
%where $K=\brac{\frac{\lambda_w}{4\pi d_0}}^2$, $\lambda_w$ is the wavelength, $d_0$ is the reference distance, $d$ is the distance between the receiver and transmitter, $\eta\geq2$ is the path-loss exponent, and $e^{\beta y}$ is the shadow fading component where $\beta = \frac{\log10}{10}$ and $y$ is a Gaussian random variable of zero mean and  $\sigma^2_{sh}$ variance. The small scale fading component is $\hbar$. We consider Rayleigh fading model. Therefore, $\hbar$ is a unit mean exponential random variable.
Let denote $\mathcal{D}(c)$ as the set of UEs on channel $c \in \mathcal{F}$.
Let $P_i$ and $\noise$ denote the transmit power of UE $i$ and noise power, respectively.
The signal-to-interference-plus-noise ratio (SINR) at the receiver of UE $i$ on channel $c$ is given as:
\begin{equation}\label{eq:sinr}
\gamma_i(c) = \frac{P_i g_i}{\sum_{j\in\mathcal{D}(c)\backslash i}P_j g_{j,i}+\noise},
\end{equation}
where $g_i$ is the channel power gain between UE $i$ and its receiver, $g_{j,i}$ is the channel power gain between UEs $i$ and $j$. 
These $g_i$ and $g_{j,i}$ take into account the path-loss, shadowing, and small-scale fading. 
The theoretical data rate $\nu_i(c)$ of UE $i$ on the channel $c$ of bandwidth $W_c$ is given by the classical Shannon capacity formula,
%\begin{equation}\label{eq:rate}
$\nu_i(c)= W_c\log_2\brac{1+\gamma_i(c)}.$
%\end{equation}

Note that the channel power gains $g_i$, $g_{i,j}$ are subject to random variations. 
These variations arise due to randomly varying channel gain, feedback errors, feedback delay errors, and quantization errors~\cite{goldsmith2005wireless}.
Therefore, all the quantities defined are in fact random variables.
We denote $\hat{\nu}$ and $\nu$ as the estimated data rate and the expected data rate, respectively.

\subsection{Problem Formulation}
Our objective is to maximize the expected sum data rate of the network by assigning channels to UEs.
Let $\bar{c} = \brac{c_i,c_{-i}}$ denotes a channel assignment vector where UE $i$ is assigned the channel $c_i\in \mathcal{F}$ and UEs other than UE $i$ are assigned the channel vector $c_{-i}\in \mathcal{F}^{\abs{\mathcal{D}}-1}$. The estimated data rate of a UE depends on vector $\bar{c}$ and is denoted as $\hat{\nu}_j\brac{\bar{c}}$. The objective function is
%\begin{equation}\label{eq:obj_func} 
$\hat{\phi}\brac{\bar{c}} = \sum_{j \in  \mathcal{D}}\hat{\nu}_j\brac{\bar{c}}.$
%\end{equation}
Formally, CAP is stated as: 
\begin{equation}
\begin{aligned}
\bar{c}^* \in &~  \underset{\bar{c}\in \mathcal{F}^{\abs{\mathcal{D}}}}{\argmax}
& & \label{eq:CAP} \phi\brac{\bar{c}},
\end{aligned}
\end{equation}
where $\phi\brac{\bar{c}}=\mathbb{E}[\hat{\phi}\brac{\bar{c}}]$ is the expected value over all the randomness.
We seek to maximize the average sum data rate by using only estimates of data rates. Hence, the above problem is a SOP~\cite{spall2005introduction}.
In the next sections, we develop a general solution framework for this kind of SOPs.

\section{Noisy Potential Game Framework}\label{sec:NPG}
In real scenarios, UEs don't experience the theoretical data rate and have access only to estimates of their average throughput that is corrupted by noise. 
%Thus their utility function becomes a random variable due to the random noise.
In order to develop a distributed solution to the CAP, we model the channel assignment problem~\eqn{eq:CAP} as a stochastic game.
\begin{defn}\label{def:cap_game}[CAP game]
A CAP game is defined by the tuple $\mathcal{\hat{G}} \coloneqq \{\mathcal{D},\{\strategy_i\}_{i\belongs \mathcal{D}},\{\hat{\utility} _i\}_{i\belongs \mathcal{D}}\}$, where $\mathcal{D}$ is a set of UEs that are players of the game, $\cbrac{\strategy_i}_{i\belongs \mathcal{D}}$ are action sets consisting of orthogonal channels, $\hat{\utility} _i:\strategy\tendsto \mathcal{R}$ are random utility functions with finite expectation, and $\strategy \coloneqq \strategy_1\times\strategy_2\times\ldots\strategy_{\abs{\mathcal{D}}}$. 
\end{defn}

An action profile $a \coloneqq \brac{a_i,a_{-i}}$ where $a_i\in X_i$ is the action of player $i$ and $a_{-i} \in X_{-i}$ is the action set of all the players except player $i$. Note that the action vector $a\in X$ is the same as the channel assignment vector $\bar{c}$ and $X = \mathcal{F}^{\abs{\mathcal{D}}}$.

Potential games are attractive class of games, using which distributed solutions to optimization problems can be designed. If the objective function of the optimization problem is aligned with the potential function, then global maximizers of the objective are also the optimal Nash Equilibria (NE) of the game. The optimal NEs are the maximizers of the potential function.
Moreover, for potential games with deterministic utilities, an NE always exists and there exist algorithms that are guaranteed to converge to NEs or to the global maximizers of the potential function~\cite{monderer_potential_1996}.
Let us thus recall the definition of potential games. 
\begin{defn}\label{def:pg}[Potential game]
A game  $\mathcal{G}\coloneqq \{\mathcal{D},\{\strategy_i\}_{i\belongs \mathcal{D}},\{\utility_i\}_{i\belongs \mathcal{D}}\}$ is a (deterministic) potential game if there is a potential function $h:\strategy\rightarrow \mathcal{R}$ such that $\forall i \in \mathcal{D}$, $\forall a_i, a_i'\in \strategy_i $ and $\forall a_{-i}\in \strategy_{-i}$,
\begin{equation}\label{eq:pg}
\utility_i(a_i,a_{-i})-\utility_i(a_i',a_{-i}) = h(a_i , a_{-i}) - h( a'_i , a_{-i}).
\end{equation}
\end{defn}

This framework cannot be directly used for our CAP game because of the random utilities. We thus propose in this paper a new class of games, 
namely {\it noisy potential games}. 
\begin{defn}\label{def:npg}[Noisy potential game]
Let the expected utility of player $i$ is denoted as $\utility _i=\mathbb{E}[\hat{\utility} _i]$.
The game $\mathcal{\hat{G}} \coloneqq \{\mathcal{D},\{\strategy_i\}_{i\belongs \mathcal{D}},\{\hat{\utility} _i\}_{i\belongs \mathcal{D}}\}$  is a noisy potential game if the game $\mathcal{G}\coloneqq \{\mathcal{D},\{\strategy_i\}_{i\belongs \mathcal{D}},\{\utility _i\}_{i\belongs \mathcal{D}}\}$ is a potential game. 
\end{defn}

We now design the utility function of the CAP game so as to obtain a noisy potential game and align the potential function with the objective function of the CAP optimization problem. 
We consider the following utility function which represents the marginal contribution of player $i$ to the global utility averaged over $N$ samples:
\begin{equation}\label{eq:mean_utility_noisy}
\hat{\utility}^N _i(a_i,a_{-i}) = \frac{1}{N}\sum_{k=1}^N\hat{\utility} _{i,k}(a_i,a_{-i}),
\end{equation} where $N$ is the number of estimation samples, and $\hat{\utility} _{i,k}$ is given by:
\begin{equation}\label{eq:utility_noisy}
\hat{\utility}_{i,k}(a_i,a_{-i}) = \sum_{j\in \mathcal{D}(a_i)}\hat{\nu}_j(a_i,a_{-i})-\sum_{j\in \mathcal{D}(a_i)\backslash i}\hat{\nu}_j(a_i,a_{-i}),
\end{equation}
where $\hat{\nu}_j$ is the measured data rate of player $j$ and 
 $\mathcal{D}(a_i) = \cbrac{j\in\mathcal{D}: a_j = a_i}$ is the set of UEs using the same channel as $i$.
Note that random utility $\hat{\utility}_{i,k}$ may have a large variance but the variance of the utility $\hat{\utility}^N _i$ can be reduced by increasing the number of samples $N$.
We will see in the next section that the number of samples $N$ must be designed carefully so as to preserve the convergence properties of potential games. 

We have the following result. The proof is straightforward.
\begin{proposition}
A CAP game $\mathcal{\hat{G}}^N  \coloneqq \{\mathcal{D},\{\strategy_i\}_{i\belongs \mathcal{D}},\{\hat{\utility}^N _i\}_{i\belongs \mathcal{D}}\}$ with utilities defined in (\ref{eq:mean_utility_noisy}), (\ref{eq:utility_noisy}) is a noisy potential game with potential function $\phi\brac{a}$.
\end{proposition}

In the rest of the paper, we consider the CAP noisy potential game $\mathcal{\hat{G}}^N$.

\section{Learning in Presence of Noise}\label{sec:learning}
In this section, we first describe the proposed binary log-linear algorithm (BLLA) for learning in the presence of noise. Then, we give the results on convergence of BLLA.

The details of BLLA are described in Algorithm~\ref{alg:bllla} and shown in~\fgr{fig:timeslotfig}. 
Each time slot is divided into two phases of size $N$ samples each (Phase I and Phase II). 
At the beginning of each time slot $t$, BS randomly selects a player $i$ and a trial action  $\hat{a}_i \in X_i$ with uniform probability. 
Also, BS informs all the players $j$ such that $a_j(t-1)\in\{a_i(t-1),\hat{a}_i\}$ to estimate their data rates and feedback this information to the BS at the end of the two phases.
Player $i$ plays action $a_i(t-1)$ and $\hat{a}_i$  during Phase I and Phase II, respectively.
At the end of Phase II, all players on $a_i(t-1)$ and $\hat{a}_i$  feedback to the BS
their two estimates of their sampled mean data rates corresponding to Phases I and II.
BS calculates the utility of player $i$ according to~\eqn{eq:mean_utility_noisy} and 
selects an action from the set $\cbrac{a_i(t-1),\hat{a}_i}$ according to~\eqn{eq:lll1},
where $\tau(t)$ is a temperature parameter that governs the convergence properties of BLLA.
Then, BS informs player $i$ with the selected action.
This feedback requires only one bit. 
BLLA is distributed in nature because only a few players have to feedback to the BS.

%\graphicspath{{Figures//}}
\begin{figure}
  \centering
\includegraphics[width=01\linewidth]{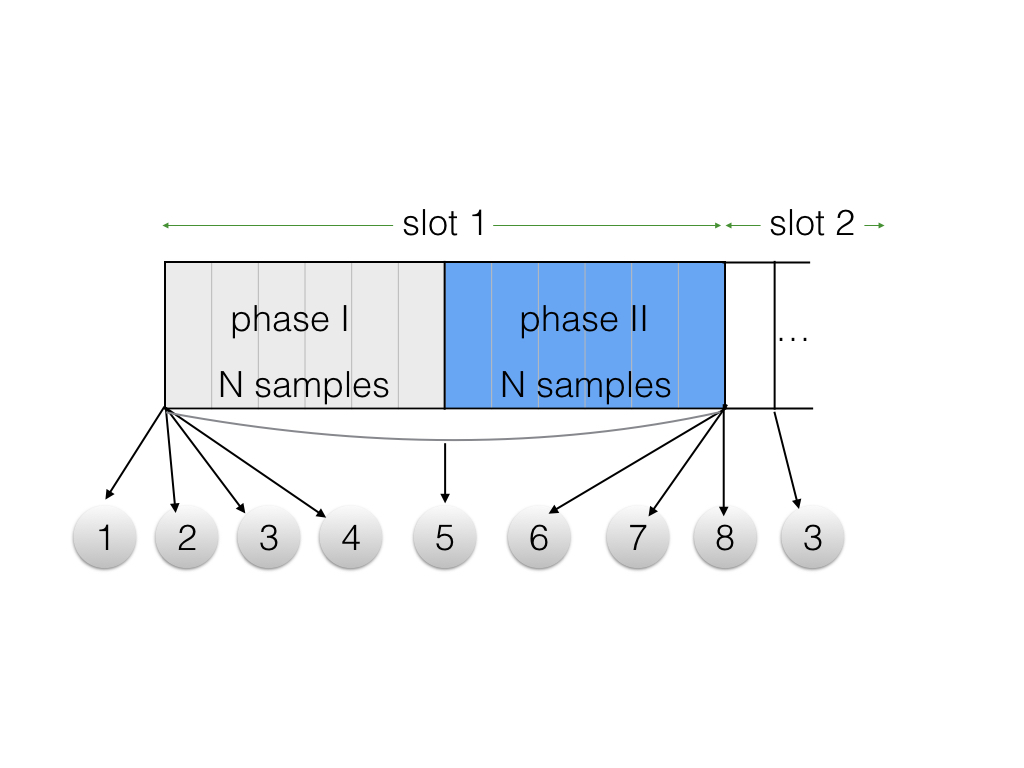}
  \caption{Time slots, phases, and steps of BLLA are shown in this figure. Steps are represented using circles. Step~5 is estimation duration of the UEs on two channels $a_i(t-1)$ and $\hat{a}_i$.}
  \label{fig:timeslotfig}
\end{figure}

\begin{algorithm}
\caption{Binary Log-linear Learning Algorithm}\label{alg:bllla}
\begin{algorithmic}[1]

\State {\bf Initialisation:} Start with arbitrary action profile  $a$.
\State {\bf While $t \geq1$ do}
\State Set parameter $\tau(t)$.
\State BS randomly selects a player $i$ and a trial action  $\hat{a}_i \in X_i$ with uniform probability. 
BS informs player $i$ and all the players with actions $a_i(t-1)$ and $\hat{a}_i$ 
to estimate their sample mean data rates.
\State Player $i$ plays action $a_i(t-1)$ and $\hat{a}_i$  during Phase I and Phase II, respectively.
\State At the end of Phase II, all players with actions $a_i(t-1)$ and $\hat{a}_i$  feedback to BS
their two estimates of their sample mean data rates corresponding to Phases I and II.
\State BS calculates $\hat{U}^N_i\brac{a(t-1)}$, $\hat{U}^N_i\brac{\hat{a}_i,a_{-i}(t-1)}$, and selects action 
$\hat{a}_i$ with probability 
\begin{equation}\label{eq:lll1}
\brac{1+e^{\Delta^N_i/\tau}}^{-1},
\end{equation}
where $\Delta^N_i = \hat{U}^N_i\brac{a(t-1)}-\hat{U}^N_i\brac{\hat{a}_i,a_{-i}(t-1)}$.
\State BS informs player $i$ to play the selected action. All the other players repeat their previous actions, i.e., $a_{-i}(t)= a_{-i}(t-1)$. 

\end{algorithmic}

\end{algorithm}

\subsection{Convergence of BLLA}
In this subsection, we present the results of convergence of BLLA for both the cases of bounded and unbounded noise.
BLLA generates an irreducible Markov chain over the action space of the CAP game $\mathcal{\hat{G}}^N$. 
However, as the parameter $\tau$ goes to zero, the stationary distribution concentrates on a few states. 
The states whose limit probability is strictly positive as $\tau$ goes to zero are called stochastically stable. 
It is known that for exact potential games the stochastically stable states of BLLA are the maximizers of the potential function~\cite{marden2010revisiting}.
 We extend this result to noisy potential games.

\begin{theorem}\label{thm:conv_blla}
The stochastically stable states of BLLA are the global maximizers of the potential function $\phi(a)$ if one of the following holds.
\begin{enumerate}
\item The estimation noise is bounded in an interval of size $\ell$ and the number of estimation samples used are
\begin{equation}\label{eq:samplesBounded}
N \geq \brac{\log\brac{\frac{4}{\xi}}+\frac{2}{\tau}}\frac{\ell^2}{2\brac{1-\xi}^2\tau^2},
\end{equation}
where $0<\xi <1$.

\item The estimation noise is unbounded with finite mean and variance.
Let $M(\theta)$ be moment generating function of noise. 
Assume that $M(\theta)$ is finite. Let $\theta^* = \argmax_{\theta}\theta\brac{1-\xi}\tau-\log\brac{M(\theta)}$.
The number of samples used are
\begin{equation}\label{eq:samplesUnbounded}
N \geq \frac{\log\brac{\frac{4}{\xi}}+\frac{2}{\tau}}{\log\brac{\frac{e^{\theta^*\brac{1-\xi}\tau}}{M(\theta^*)}}}.
\end{equation}
\end{enumerate}
\end{theorem}
\begin{IEEEproof}
See Appendix~\ref{app:conv_blla}.
\end{IEEEproof}

%The following corollary follows immediately for the case when the noise is assumed to have Gaussian distibution.
%\begin{corollary}\label{cor:normal}
%Let the noise has standard normal probability distribution. 
%Then the stochastically stable states of BLLA are the global maximizers of the potential function if 
%\begin{equation}\label{eq:samples}
%N \geq \frac{2\log\brac{\frac{4}{\xi}}+\frac{4}{\tau}}{\tau^2\brac{1-\xi}^2}.
%\end{equation}    
%\end{corollary}

A small $N$ is desired for practical implementations. We choose the lowest $N$ that satisfies Theorem~\ref{thm:conv_blla}. 
In Theorem~\ref{thm:conv_blla}, we have a convergence in probability for fixed parameter $\tau$. In Theorem~\ref{thm:decTau}, we consider the case of decreasing parameter $\tau$ for which we obtain an almost sure convergence to optimal state as in simulated annealing with cooling schedule~\cite{hajek1988cooling}.
\begin{theorem}\label{thm:decTau}
  Consider BLLA with a decreasing parameter $\tau(t) = 1 / \log(1+t)$, and the number of samples $N(\tau)$ is given by Theorem~\ref{thm:conv_blla}. Then, BLLA converges with probability $1$ to the global maximizer of the potential function. 
\end{theorem}

\begin{proof}
  See Appendix~\ref{app:conv_blla_decreasing}.
\end{proof}

\section{Simulations}\label{sec:simulations}
\begin{table}\label{tbl:sim_parameters}
\caption{Simulation parameters.} % title of Table
\centering % used for centering table
\begin{tabular}{l c l c l c |} % centered columns (2 columns)
\hline %inserts double horizontal lines
\bf Parameter & \bf Variable &\bf Value \\  % inserts table
%heading
\hline % inserts single horizontal line
 % inserting body of the table
%Network radius & $\mathcal{L}$ & $200$ \\ \hline
%System bandwidth & $W$ & 5~MHz\\ \hline 
Number of orthogonal channels & $\mathcal{F}$ & 5\\ \hline
Channel bandwidth & $W_c$ & 180~KHz\\ \hline 
Carrier frequency & $f_c$ & 2.6~GHz \\ \hline
%Max data rate of UE alone in cell & $\nu_{\max}$ & 100~Mbps\\ \hline 
%Number of BSs & $ \NumBSs $ & 1 \\ \hline
%Number of cellular UEs & $ \abs{\mathcal{D}_c} $ & 5 \\ \hline
Number of UEs & $ \abs{\mathcal{D}} $ & 20 \\ \hline
Total transmit power of BS &$ P_{\text{BS}}$ & 46 dBm\\ \hline
Transmit power of UE &$ P_{\text{UE}}$ & 25 dBm\\ \hline
Minimum SINR &$\gamma_{\min}$ & -10 dB\\ \hline
Maximum SINR &$\gamma_{\max}$ & 23 dB\\ \hline

Additive noise power per channel &$\noise$ & $-174+10\log\brac{W_i}$ dBm\\ \hline
Path-loss exponent &$ \eta$ & 3.5\\ \hline
Shadowing variance & $\sigma_{sh}$ & 6 \\ \hline
\end{tabular}
\label{tbl:sim_para} % is used to refer this table in the text
\end{table}

\subsection{Simulation Parameters}

In this Section, we present simulation results on the downlink considering standard wireless system parameters shown in~\tbl{tbl:sim_para}. 
We consider that a BS is located at the center of a region of radius $200$~m.
Among 20 UEs there are 5 UECs.
The UECs have dedicated channels and no two UECs are on the same channel.
These UECs serve as passive players of the game because they do not change their channel.
The receivers of UED transmitters are located around them uniformly random over a region of radius $20$~m.
The UEDs learn their channel with the help of the BS and hence are the active players of the game.

The variations of Rayleigh fading over time are considered as the noise component for all the simulations. 
Channel variations result in random UEs data rates. The data rates are bounded because the SINR is bounded between $\gamma_{\min}$ and $\gamma_{\max}$. We assume only bounded noise in the simulations where the noise interval is set to $\ell=1$ due to normalized utilities. Besides, the additive white Gaussian noise with power $\noise$ is considered. 

\subsection{Simulation Results}
%In this subsection, we show the simulations for the downlink of the D2D cellular network.
In~\fgr{fig:conv_FixTau}, convergence to the maximum sum data rate of BLLA is shown with fixed temperature $\tau=0.1$ and decreasing temperature $\tau(t) = 0.1/\log\brac{1+t}$\footnote{Note that $\tau(t) = 0.1/\log\brac{1+t}$ works well even though it is smaller than that given by Theorem~\ref{thm:decTau}. The reason being that the height of the highest local maximum is smaller than the global maximum~\cite{hajek1988cooling}. We consider that height to be $10\%$ of the global maximum, which is reasonable.}. 
The number of samples are calculated according to~\eqn{eq:samplesBounded} for $\tau = 0.1$ and $\xi = 10^{-5}$. BLLA reaches the maximum sum data rate with both fixed and decreasing temperatures. However, it has more variations for fixed temperature. For decreasing temperature, the probability of staying at the maximum is higher.%, which is evident from~\fgr{fig:conv_FixTau}.
%\graphicspath{{Figures//}}
\begin{figure}[t]
  \centering
\includegraphics[width=0.9\linewidth]{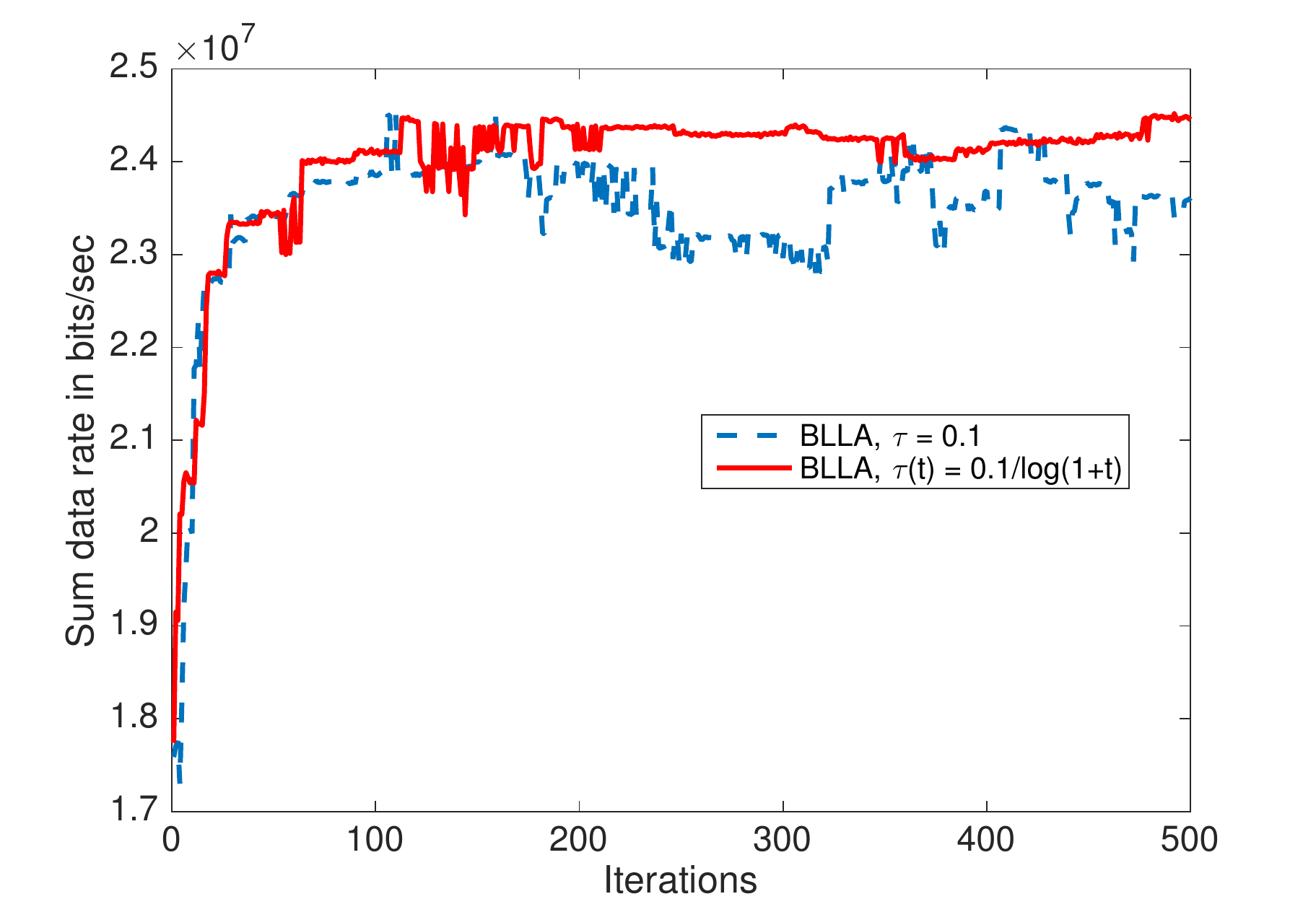}
  \caption{Convergence of BLLA for fixed temperature and decreasing temperature.}
  \label{fig:conv_FixTau}
\end{figure}

To study the effect of temperature we show the performance of BLLA for different temperatures in~\fgr{fig:conv_diff_FixTau}. As before, the samples are calculated corresponding to different temperatures. 
For higher temperature $\tau = 0.5$, BLLA exhibits huge variations. 
The probability of being at a local maximum decreases with increasing temperature. As temperature decreases, the variations also decrease. Also, the probability of being at a local maximum increases with decreasing temperature. Therefore, the temperature should be chosen carefully to obtain the desired performance. BLLA with smaller $\tau = 0.05$ gives the desired performance. 

\begin{figure}[t]
  \centering
\includegraphics[width=0.9\linewidth]{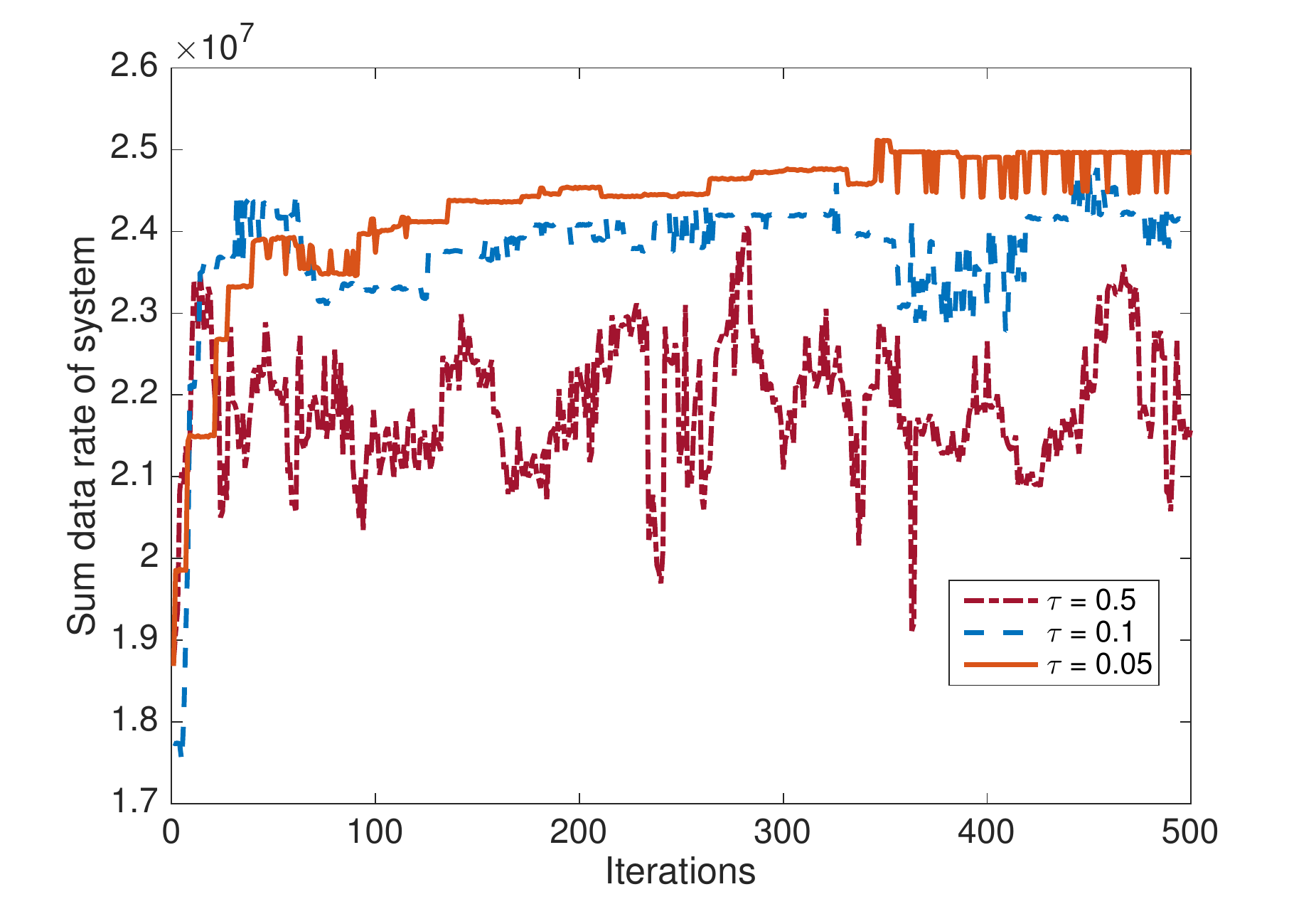}
  \caption{Convergence of BLLA for different fixed temperatures.}
  \label{fig:conv_diff_FixTau}
\end{figure}

We study the effect of the number of samples on the performance of BLLA  in~\fgr{fig:conv_diff_samples}.
If players take decisions after every single sample, i.e., if the estimation errors are ignored, BLLA exhibits large variations due to noise.
As the number of samples increases the performance of BLLA improves. 
If the number of samples is taken according to Theorem~\ref{thm:conv_blla} then BLLA provides high and stable sum data rate.
\begin{figure}[t]
  \centering
\includegraphics[width=0.9\linewidth]{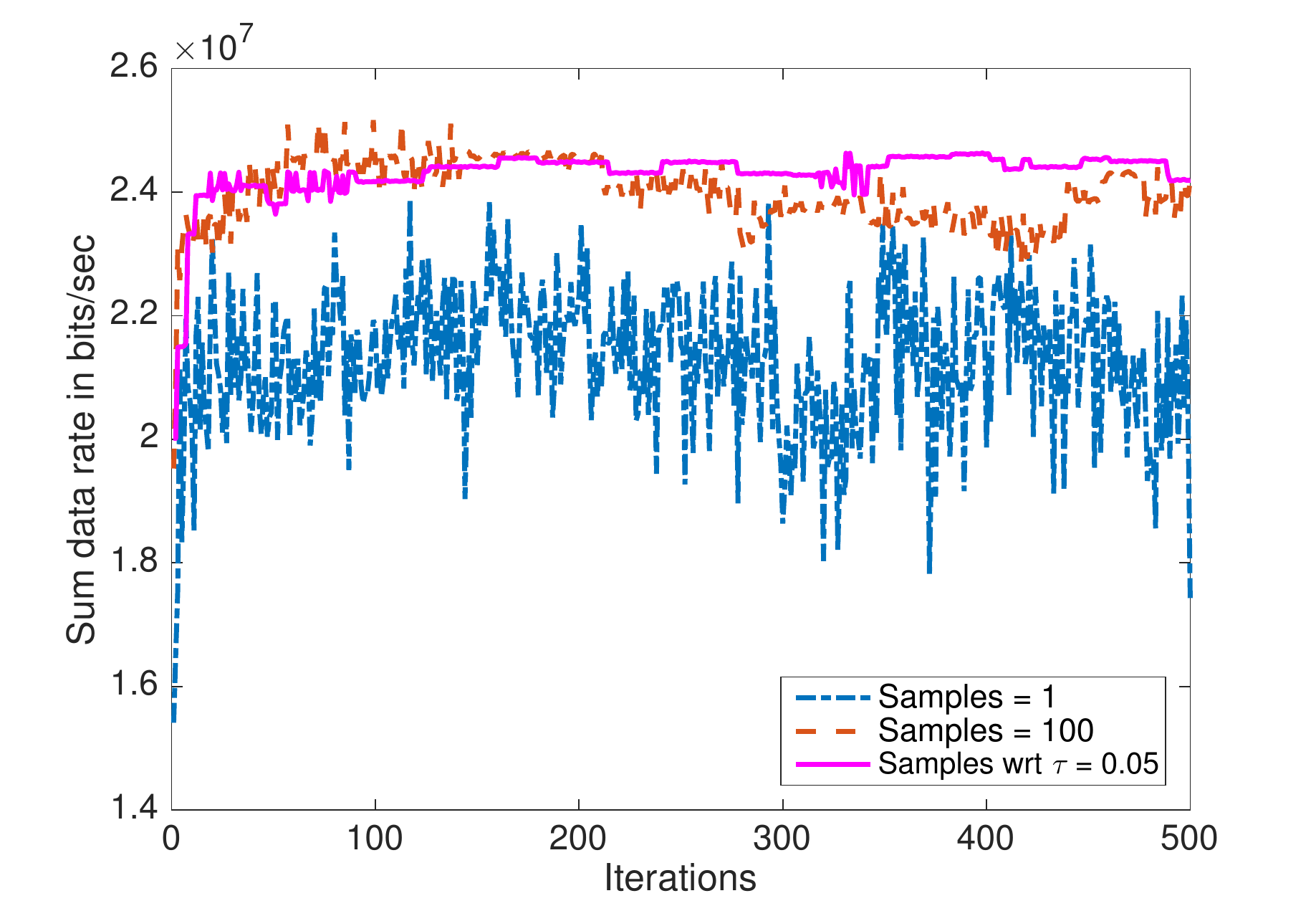}
  \caption{Effect of the number of samples on the convergence of BLLA.}
  \label{fig:conv_diff_samples}
\end{figure}

\begin{figure}[t]
  \centering
\includegraphics[width=0.9\linewidth]{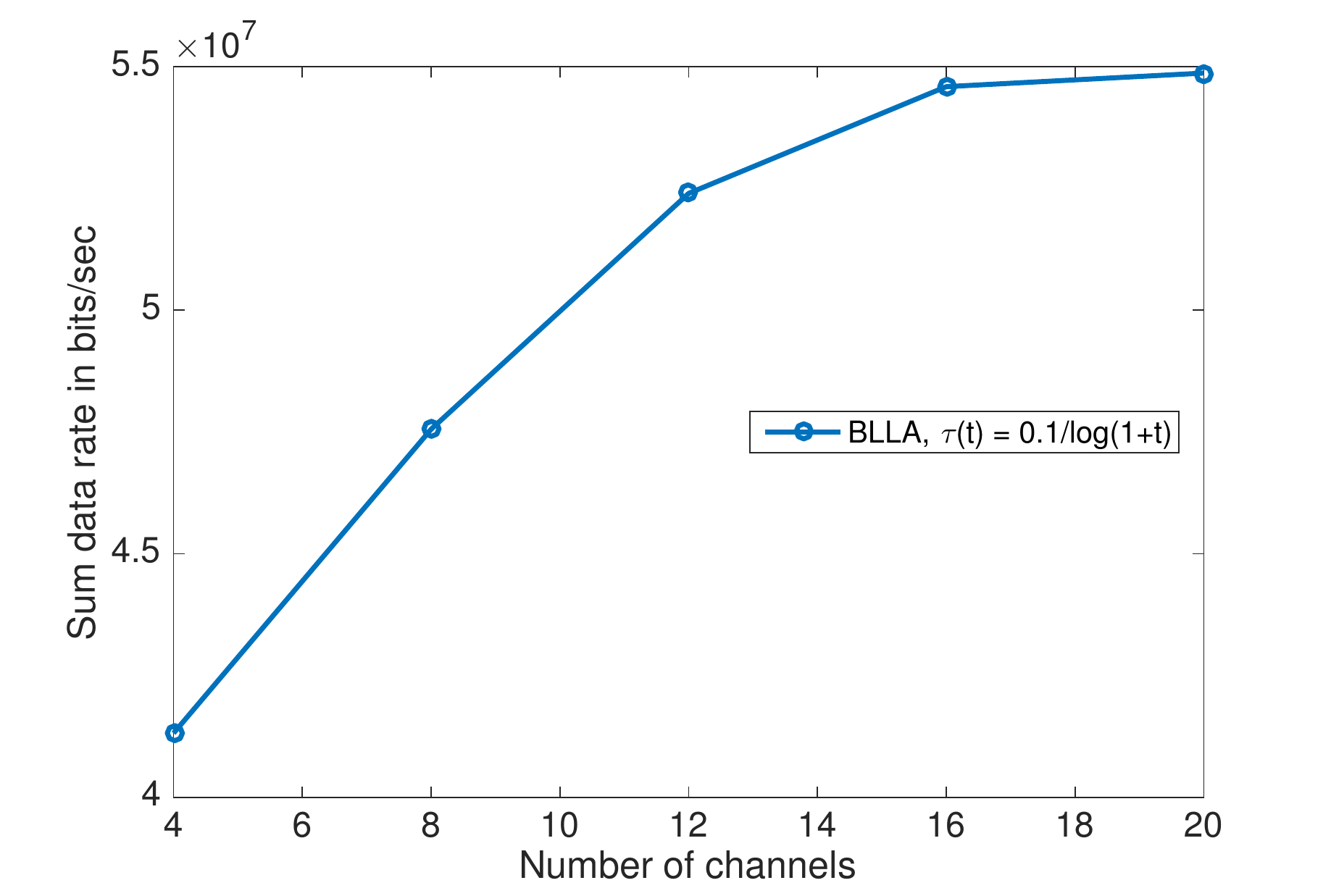}
  \caption{Effect of the number of channels on the sum data rate. }
  \label{fig:channnel_vs_throughput}
\end{figure}

We now study the effect of the number of channels in~\fgr{fig:channnel_vs_throughput}.
This plot shows the average sum data rate obtained from BLLA  at the end of 500  iterations averaged over 1000 realizations.
We see that the sum data rate increases as the number of channels increases. This is intuitive because the optimal channel assignment has lower interference per channel. As evident from the figure, BLLA correctly tracks this phenomenon.
\begin{figure}[t]
  \centering
\includegraphics[width=0.9\linewidth]{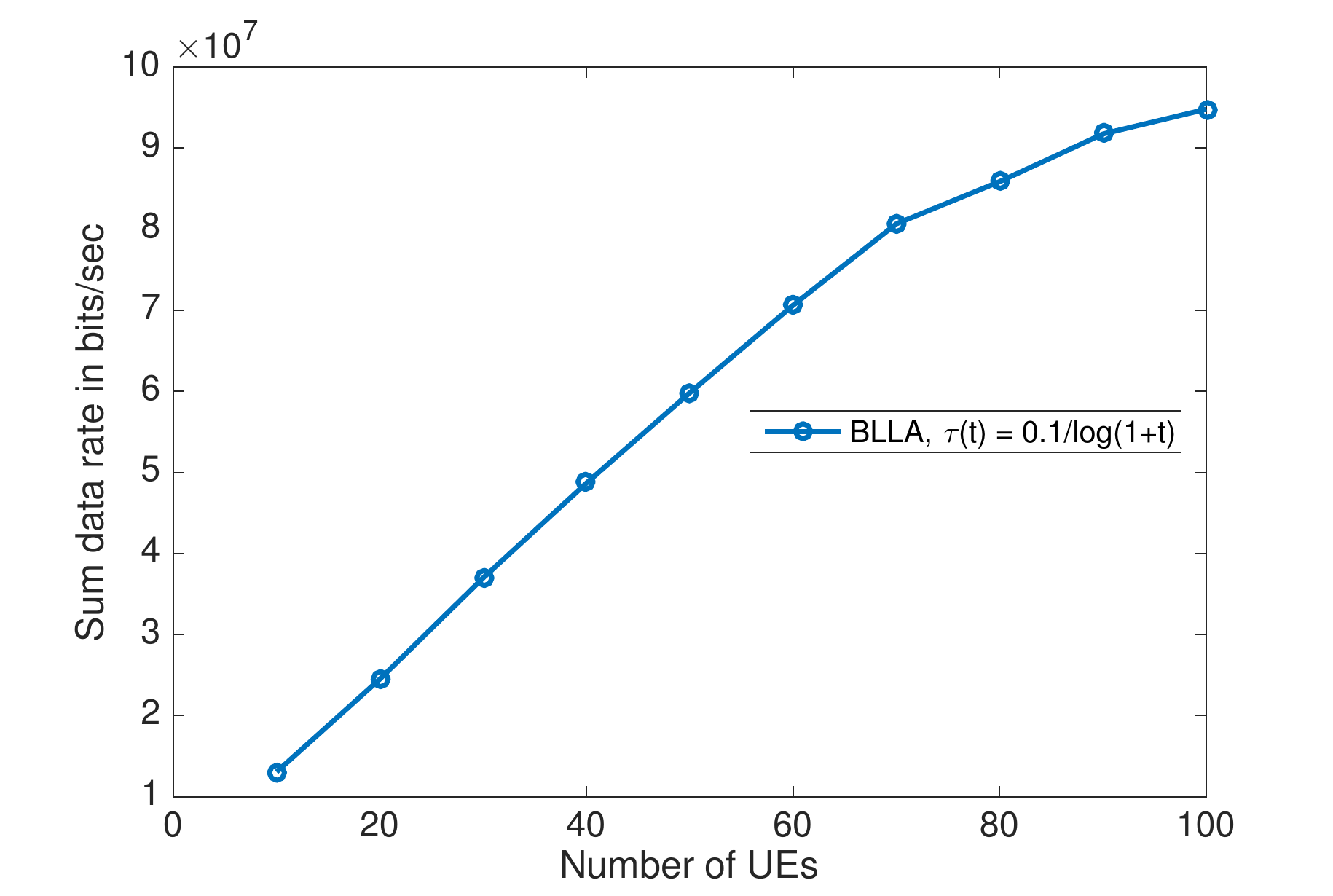}
  \caption{Effect of the number of UEs on the sum data rate.}
  \label{fig:UEs_vs_throughput}
\end{figure}

We also study the performance of BLLA by varying the number of UEs in~\fgr{fig:UEs_vs_throughput} for 10 orthogonal channels and 10 UECs. 
As before, the sum data rate is obtained from BLLA  at the end of 500 iterations averaged over 1000 realizations. As the number of UEs increases (up to approximately 60 in the figure), the sum data rate first linearly increases because of the increasing traffic. A linear growth is observed as long as interference is controlled. Sixty is much larger than the number of available channels, which means that BLLA manages to assign frequencies in such a way that UEs do not interfere too much. After 60 UEs, the increase is reduced because interference significantly affects the sum data rate.
BLLA exactly tracks this behaviour as evident from~\fgr{fig:UEs_vs_throughput}.

\begin{figure}[t]
  \centering
\includegraphics[width=0.9\linewidth]{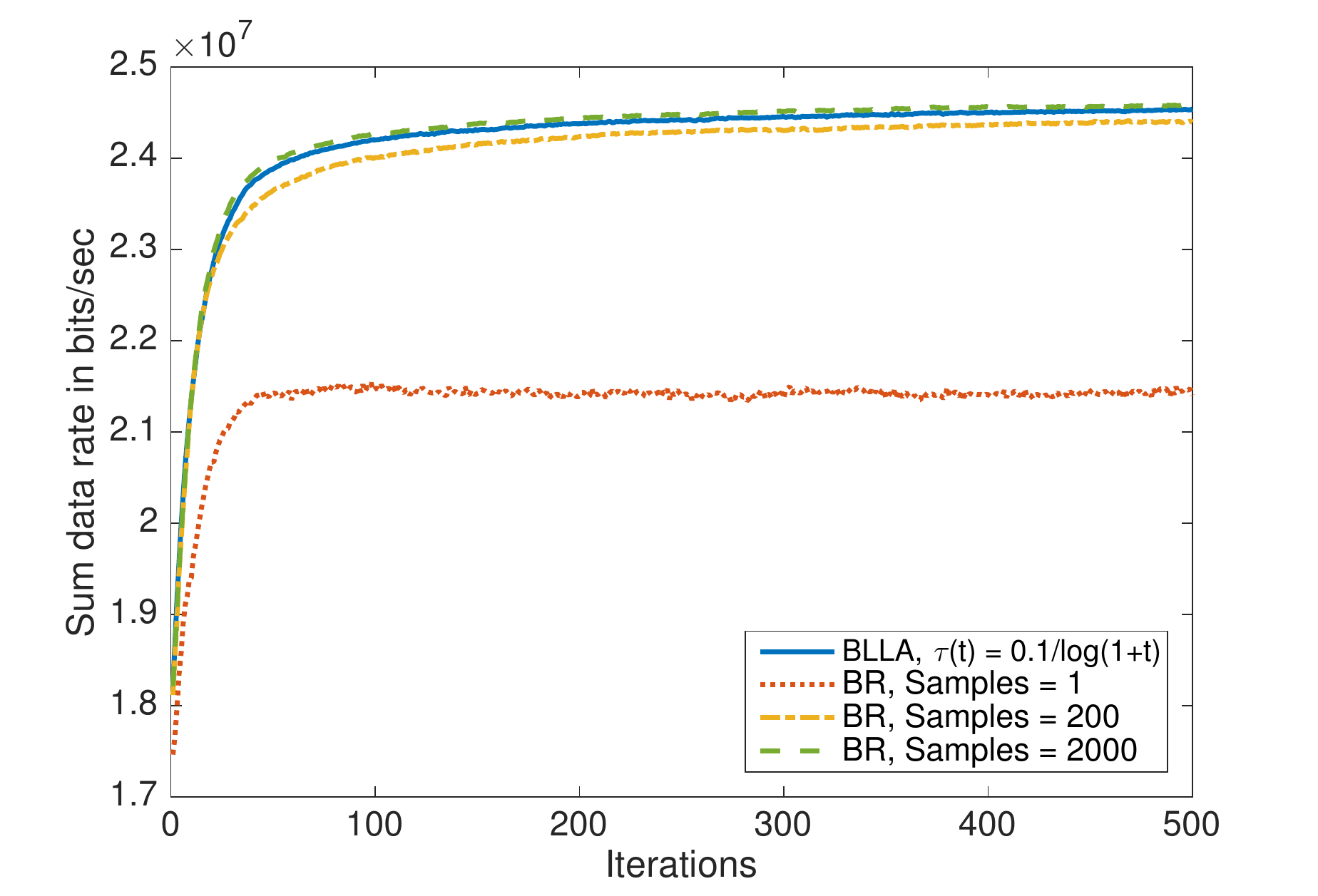}
  \caption{Comparison of BLLA and better response algorithms.}
  \label{fig:compare_blla_br}
\end{figure}
We now compare in~\fgr{fig:compare_blla_br} the performance of BLLA and better response (BR) algorithm, 
which accepts the trial action only if its utility is better than the current action.
Best response algorithm, which is same as better response with two actions, was applied to CAP with the objective of minimizing the total interference in Wireless Sensor
Networks (WSN) in~\cite{chen2011game}.
The parameter of BLLA is $\tau(t) = 0.1/\log\brac{1+t}$.
Each curve in the figure is obtained by averaging over 1000 realizations of the algorithms.
BR performs the worst when noise is not taken into account, which corresponds to one estimation sample case.
When the number of samples is increased to $200$, BR improves. However, BLLA is better than BR.
For 2000 samples, BR performance is the same as BLLA. 
It shows that the number of samples for BR has to be tuned carefully to obtain the desired performance.
On the contrary, BLLA performs better with the fixed number of samples without any tuning.
Note that theoretically BR needs an infinite number of samples. There is no theoretical guarantee for its convergence for a finite number of samples.
On the contrary, as we have proved, BLLA has a theoretical guarantee of convergence with a finite number of samples.

%\subsection{Uplink Sum Data Rate}
%\begin{figure}
%  \centering
%\includegraphics[width=0.9\linewidth]{fixedTau_vs_VaryingTau}
%  \caption{BLLA for uplink with fixed $\tau = 0.05$ and $\tau(t) = 0.1/\log(1+t)$.}
%  \label{fig:compare_fixTau_varyingTau_ul}
%\end{figure}
%
%We also show that BLLA can be applied for channel assignment in the uplink of the network.
%The uplink sum data rate obtained from BLLA for fixed and decreasing temperatures are shown in~\fgr{fig:compare_fixTau_varyingTau_ul}. 
%With fixed temperature BLLA reaches the global maximum but leaves it quickly. 
%However, with decreasing temperature BLLA stays at the global maximum for most of the time.
%This shows the applicability of BLLA for both uplink and downlink.
%

\section{Conclusions}\label{sec:conclusions}
A novel optimal solution for CAP in D2D wireless networks that takes into account throughput estimation noise is presented.
To capture this noise, a noisy potential game framework is introduced.
A distributed Binary Log-linear Learning Algorithm that achieves the optimal channel assignments is proposed.
 A sufficient number of estimation samples that guarantees the convergence in both cases of bounded noise and unbounded noise are given that are validated using simulations.
BLLA achieves the optimal sum data rate. 
The sum data rate increases with the number of channels and with the number of UEs.
BLLA performs better than better response algorithm.
%*************************************************************
%*****    BIBLIOGRAPHY
%*************************************************************
%\bibliographystyle{siam}
\bibliographystyle{ieeetr}
\bibliography{shabb_cre_model_MC}
%\bibliography{bibJournalList.bib}
%\bibliographystyle{unsrtnat}
%\bibliographystyle{Classes/jmb} % bibliography style
%\bibliography{/Bibtex/bibJournalList,Bibtex/book,Bibtex/standard,Bibtex/shabb_cre_model,Bibtex/}

\appendix

\subsection{Proof of Convergence of BLLA with Fixed~$\tau$}\label{app:conv_blla}

If the utilities of the CAP game are deterministic and without noise then the CAP game becomes an exact potential game. For an exact potential game, the stochastically stable states of BLLA are the maximisers of potential~\cite{marden2010revisiting,leslie2011equilibrium,Shabb_EPG_2016}. 
We prove the same even for a noisy potential CAP game $\mathcal{\hat{G}}^N$ if the number of samples is chosen carefully.
Our proof approach is to show that for a particular number of samples the resistance BLLA with estimated utilities is same as that of with the deterministic utilities\footnote{In all the proofs, the considered utilities are normalized by the maximum potential $\phi_{\max}$.}.
This kind of proof idea based on resistance is similar to that of~\cite{Young93,leslie2011equilibrium}.

We first recall the computation of the resistance of BLLA in a deterministic potential game, as in~\cite{marden2010revisiting}.
Let consider the CAP game $\mathcal{G} \coloneqq \cbrac{\mathcal{D},\cbrac{\strategy_i}_{i\belongs \mathcal{D}},\cbrac{\utility_i}_{i\belongs \mathcal{D}}}$, with expected utilities $\utility _i=\expect{\hat{\utility}^N_i}$. It is an exact potential game.
BLLA induces a regular Markov process over the action space $X$ of $\mathcal{G} $~\cite{marden2010revisiting,leslie2011equilibrium,Shabb_EPG_2016}. 
Let denote $P^\tau$ as the transition matrix of the regular Markov process.

\begin{defn}[Resistance of transition~\cite{marden2010revisiting}]\label{def:resistance_classic}
Let $a= \brac{a_i,a_{-i}}$ and $b= \brac{a'_i,a_{-i}}$ be action profiles such that only player $i$ changes its action.
Let $P_{\text{ab}}^\tau$ be a strictly positive probability transition function with the parameter $\tau$.
A non negative number $R_{\text{ab}}$ is the resistance of transition $a\to b$  if
\begin{equation}\label{eq:resistance_def_literature}
 0< \lim_{\tau \to 0^{+}} \frac{P_{\text{ab}}^\tau}{e^{-\frac{R_{\text{ab}}}{\tau}}}<\infty.
\end{equation}
\end{defn}

%In general, it is difficult to compute resistance from the above definition. Moreover, it does not include many cases of functions. For example, the resistance is not defined for $P_{\text{ab}}^\tau = \tau$.

To develop easy rules to compute the resistance of a function we give a generalised definition of resistance in the following.
Let $``o"$ and $``\omega"$ denote little "o" order and little omega order, respectively.
We call function $g(\tau)$ a sub-exponential function that satisfies $g\in o\brac{e^{k/\tau}}$ and $g \in \omega\brac{e^{-k/\tau}}$ for any $k>0$.
%for any positive function we allow for negative resistance in order 
\begin{defn}[Resistance of positive function]\label{def:resistance_pos_func}
Let $f(\tau)$ be a strictly positive function. If there is a sub-exponential function $g(\tau)$ and a number $R$ such that~\eqn{eq:res_pos_func} holds, then $R$ is unique (see Lemma~\ref{lem:Res_unique}) and is called the resistance of $f$, denoted by $Res(f)$.
%The resistance of a strictly positive $f(\tau)$ is $R$ and is denoted as $\text{Res}(f)$
%if there exist a strictly positive function $g(\tau)$ and a real number $R$ such that 
%$g\in o\brac{e^{k/\tau}}$ and $g \in \omega\brac{e^{-k/\tau}}$ for any $k>0$;
%and 
\begin{equation}\label{eq:res_pos_func}
\lim_{\tau \to 0} \frac{f(\tau)}{g(\tau)e^{-\frac{\text{R}}{\tau}}} = 1.
\end{equation}
\end{defn}
\begin{remark}
Remark that Definition~\ref{def:resistance_pos_func} includes Definition~\ref{def:resistance_classic}, in which $g(\tau) = \kappa,  0<\kappa <\infty$.
\end{remark}

\begin{remark}
Note that~\eqn{eq:res_pos_func} is equivalent to
\begin{equation}\label{eq:f_equiv_Res}
f(\tau)= g(\tau) e^{-\frac{\text{Res}(f)}{\tau}} +h(\tau),
\end{equation}
where $h(\tau) \in o\brac{g(\tau)e^{-\frac{\text{Res}(f)}{\tau}}}$.
\end{remark}
\begin{lemma}\label{lem:order_g_1_g2}
Consider any two sub-exponential functions $g_1(\tau)$ and $g_2(\tau)$. Consider two real numbers $R_1$ and $R_2$. 
If $R_1<R_2$ then 
\begin{equation}
g_2(\tau)e^{-R_2/\tau} \in o \brac{g_1(\tau)e^{-R_1/\tau}}.
\end{equation}
\end{lemma}
\begin{IEEEproof}
Let $k$ be a real number. Then
\begin{equation}
\lim_{\tau\to 0}\frac{g_2(\tau)e^{-R_2/\tau}}{g_1(\tau)e^{-R_1/\tau}} =
 \lim_{\tau\to 0}\frac{g_2(\tau)}{e^{-\brac{R_2-k}/\tau}} \sbrac{\frac{g_1(\tau)}{e^{-\brac{R_1-k}/\tau}}}^{-1}.
\end{equation}
The above limit goes to zero when we choose $R_1<k<R_2$.
\end{IEEEproof}
\begin{lemma}\label{lem:Res_unique}
If $\text{Res}(f)$ exists then it is unique.
\end{lemma}
\begin{IEEEproof}
Assume that function $f$ have two different resistances $R_1$ and $R_2$.
Then, from~\eqn{eq:f_equiv_Res} we have
\begin{equation}
f(\tau) = g_1(\tau) e^{-\frac{R_1}{\tau}} +h_1(\tau)= g_2(\tau) e^{-\frac{R_2}{\tau}} +h_2(\tau),
\end{equation}
where $h_1(\tau) \in o\brac{g_1(\tau)e^{-\frac{R_1}{\tau}}}$ and $h_2(\tau) \in o\brac{g_2(\tau)e^{-\frac{R_2}{\tau}}}$.
Let $R_1<R_2$. Using Lemma~\ref{lem:order_g_1_g2}, we have $h_2 \in o\brac{g_1(\tau)e^{-\frac{R_1}{\tau}}}$.
Rearranging the term in above equation, we have
\begin{equation}
1 +\frac{h_1(\tau)}{g_1(\tau) e^{-\frac{R_1}{\tau}}}= \frac{g_2(\tau)e^{-\frac{R_2}{\tau}}}{g_1(\tau) e^{-\frac{R_1}{\tau}}} +\frac{h_2(\tau)}{g_1(\tau) e^{-\frac{R_1}{\tau}}}.
\end{equation}
Using Lemma~\ref{lem:order_g_1_g2}, we arrive at contradiction that $1=0$.

\end{IEEEproof}

%Note that the resistance can be defined for any strictly positive function $f(\tau)$. We denote it as $\text{Res}(f)$. 
%In particular $\text{Res}(P_{\text{ab}}^\tau) = R_{\text{ab}}$.
%Definition~\ref{def:resistance_classic} implies that there exist $\kappa>0$ such that
%\begin{equation}
%\lim_{\tau \to 0} \frac{f(\tau)}{e^{-\frac{\text{Res}(f)}{\tau}}} = \kappa.
%\end{equation}
%This implies that 
%\begin{equation}
%f(\tau)= \kappa e^{-\frac{\text{Res}(f)}{\tau}} +h_1(\tau),
%\end{equation}
%where $h_1(\tau) \in o\brac{e^{-\frac{\text{Res}(f)}{\tau}}}$.
%Note that $\text{Res}(f)$ can also take negative values. Note also that  if $\text{Res}(f)$ exists then it is unique.

The following Lemma gives useful rules for computing $\text{Res}(f)$.
\begin{lemma}\label{lem:comp_func}
Let $f_1$ and $f_2$ be strictly positive real valued functions. 
Let $\kappa$ be a positive constant.
 If $\text{Res}(f_1)$ and $\text{Res}(f_2)$ exist then
 \begin{enumerate}%[label=\Roman*]
% \begin{align}
% \label{eqn:Res_equivalence} f(\tau) &\sim e^{-\frac{\text{Res}(f)}{\tau}}, \mbox{if } \text{Res}(f) \neq 0, \text{Res}(f) <\infty ,\\
\item \label{eq:Res_const}$f(\tau)$ is sub-exponential if and only if $\text{Res}(f)=0$. In particular $\text{Res}(\kappa)  = 0$,
\item  \label{eq:Res_exp}$ \text{Res}(e^{-\kappa /\tau})  = \kappa$,
\item \label{eq:Res_sum_f_g}$\text{Res}(f_1 + f_2)  = \min\cbrac{\text{Res}(f_1),\text{Res}(f_2)}$,
\item \label{eq:Res_sum_f_minus_g}$\text{Res}(f_1 - f_2)  = \text{Res}(f_1), \text{if } \text{Res}(f_1)<\text{Res}(f_2)$,
\item \label{eq:Res_prod_fg}$\text{Res}(f_1f_2) = \text{Res}(f_1)+\text{Res}(f_2)$,
\item  \label{eq:Res_f_inv}$ \text{Res}(\frac{1}{f})  = - \text{Res}(f), \text{if Res}(f)\neq 0$.
\item  \label{eq:Res_mono} If $f_1(\tau) \leq f_2(\tau)$, $\text{Res}(f_1)$ and $\text{Res}(f_2)$ exist then $\text{Res}(f_2)\leq \text{Res}(f_1)$.
\item  \label{eq:Res_mono_existence} Let $f_1(\tau) \leq f(\tau) \leq f_2(\tau)$. 
If $\text{Res}(f_1)=\text{Res}(f_2)$  then $\text{Res}(f)$ exists and $\text{Res}(f)= \text{Res}(f_1)$.
% \end{align}
\end{enumerate}
 \end{lemma}

\begin{remark}
In Rule~\ref{eq:Res_sum_f_minus_g}, if $\text{Res}(f_1)=\text{Res}(f_2)$ then we cannot compute $\text{Res}(f_1 - f_2)$ because the difference of sub-exponential functions may not be a sub-exponential function.
\end{remark}

\begin{remark}
In Rule~\ref{eq:Res_mono_existence}, in general if $f_1(\tau) \leq f(\tau) \leq f_2(\tau)$ and $\text{Res}(f_1)\neq \text{Res}(f_2)$  then $\text{Res}(f)$ may not exists.
\end{remark}

\begin{IEEEproof}
\emph{Proof of rule~\ref{eq:Res_const}:}
Let $f(\tau)$ be a sub-exponential function. Choosing  $g(\tau)=f(\tau)$ from~\eqn{eq:res_pos_func} we have
\begin{equation}
\lim_{\tau \to 0} e^{-\frac{\text{Res}(f)}{\tau}}=1.
\end{equation}
Therefore, we have  $\text{Res}(f)=0$.

Assume $\text{Res}(f)=0$. From~\eqn{eq:f_equiv_Res}, we have $f(\tau)= g(\tau) +h(\tau)$, which is a sub-exponential function.

Let $f(\tau) = \kappa$ and  $g(\tau) = \kappa$ then $g(\tau) \in o\brac{e^{\frac{ \kappa}{\tau}}}$ and  $g(\tau) \in \omega\brac{e^{-\frac{ \kappa}{\tau}}}$, $ \kappa>0$. 
%By Definition~\ref{def:resistance_pos_func}, we have 
%\begin{equation}
%\lim_{\tau\to0}\frac{f(\tau)}{g(\tau)e^{-\frac{\text{Res}(f)}{\tau}}} = \lim_{\tau\to0}e^{\frac{\text{Res}(f)}{\tau}} .
%\end{equation}
Substituting these in~\eqn{eq:res_pos_func} we have $\text{Res}(\kappa)=0$.

\emph{Proof of rule~\ref{eq:Res_exp}:}
Substituting  $f(\tau) = e^{-\kappa/\tau}$ and $g(\tau) = 1$  in~\eqn{eq:res_pos_func} we get $\text{Res}(e^{-\kappa/\tau})=\kappa$.

\emph{Proof of rule~\ref{eq:Res_sum_f_g}:}
%The equation~\eqn{eq:Res_sum_f_g} is obtained as follows:
Consider that $\text{Res}(f_1)$ and $\text{Res}(f_2)$ be the resistances of functions $f_1$ and $f_2$, respectively. We have
\begin{align}
f_1(\tau)&= g_1(\tau) e^{-\frac{\text{Res}(f_1)}{\tau}} +h_1(\tau),\\
f_2(\tau)&= g_2(\tau) e^{-\frac{\text{Res}(f_2)}{\tau}} +h_2(\tau),
\end{align}
where $h_1(\tau) \in o\brac{g_1(\tau) e^{-\frac{\text{Res}(f_1)}{\tau}}}$, $h_2(\tau) \in o\brac{g_2(\tau) e^{-\frac{\text{Res}(f_2)}{\tau}}}$.
%\begin{equation}
%f_1(\tau)+f_2(\tau) 
%= g_1(\tau) e^{-\frac{\text{Res}(f_1)}{\tau}} +h_1(\tau)
%+g_2(\tau) e^{-\frac{\text{Res}(f_2)}{\tau}} +h_2(\tau),
%\end{equation}

The sum of two functions can be written as
\begin{multline}
f_1(\tau)+f_2(\tau)  = g_1(\tau) e^{-\frac{\text{Res}(f_1)}{\tau}} \left(1 +\frac{h_1(\tau)}{g_1(\tau) e^{-\frac{\text{Res}(f_1)}{\tau}} }
\right.  \\+ \left. \frac{g_2(\tau) e^{-\frac{\text{Res}(f_2)}{\tau}}}{g_1(\tau) e^{-\frac{\text{Res}(f_1)}{\tau}} } 
+ \frac{h_2(\tau)}{g_1(\tau) e^{-\frac{\text{Res}(f_1)}{\tau}} }\right),
\end{multline}
Consider the case when $\text{Res}(f_1)<\text{Res}(f_2)$. 
Using Lemma~\ref{lem:order_g_1_g2} we have $h_2 \in o\brac{g_1(\tau) e^{-\frac{\text{Res}(f_1)}{\tau}}}$. Therefore, 
\begin{equation}
f_1(\tau)+f_2(\tau) 
= g_1(\tau) e^{-\frac{\text{Res}(f_1)}{\tau}}  + h_3(\tau),
\end{equation} 
where $h_3(\tau) \in o\brac{g_1(\tau) e^{-\frac{\text{Res}(f_1)}{\tau}} }$.
According to~\eqn{eq:f_equiv_Res}, we have $\text{Res}(f_1 + f_2)  =\text{Res}(f_1)= \min\cbrac{\text{Res}(f_1),\text{Res}(f_2)}$. 

The case of $\text{Res}(f_1)=\text{Res}(f_2)$ leads to the same result as shown below.
\begin{equation}
f_1(\tau)+f_2(\tau) 
= e^{-\frac{\text{Res}(f_1)}{\tau}}\sbrac{g_1(\tau) +g_2(\tau) }  + h_1(\tau)+h_2(\tau).
\end{equation} 
Note that sum of sub-exponential functions $g_1(\tau) +g_2(\tau)$ is a sub-exponential function.
 Observe that $h_1(\tau)+h_2(\tau) \in o\brac{\sbrac{g_1(\tau) +g_2(\tau) }e^{-\frac{\text{Res}(f_1)}{\tau}}}$. 
As in the previous case, according to~\eqn{eq:f_equiv_Res} we have  $\text{Res}(f_1 + f_2)  =\text{Res}(f_1)$

\emph{Proof of rule~\ref{eq:Res_sum_f_minus_g}:}
Also, it can be shown similarly to the proof of rule~\ref{eq:Res_sum_f_g} that if $\text{Res}(f_1)<\text{Res}(f_2)$ then
\begin{equation}
\text{Res}(f_1 - f_2)  =\text{Res}(f_1).
\end{equation}

\emph{Proof of rule~\ref{eq:Res_prod_fg}:}
%The equation~\eqn{eq:Res_prod_fg} is obtained as follows:
\begin{multline}
\lim_{\tau \to 0} \frac{f_1(\tau)}{g_1(\tau)e^{-\frac{\text{Res}(f_1)}{\tau}}}
\lim_{\tau \to 0} \frac{f_2(\tau)}{g_2(\tau)e^{-\frac{\text{Res}(f_2)}{\tau}}}\\
=  \lim_{\tau \to 0} \frac{f_1(\tau)f_2(\tau)}{g_1(\tau)g_2(\tau)e^{-\frac{\text{Res}(f_1)+\text{Res}(f_1)}{\tau}}}=1.
\end{multline} 
Therefore,  $\text{Res}(f_1f_2)  =\text{Res}(f_1)+\text{Res}(f_2)$.

\emph{Proof of rule~\ref{eq:Res_f_inv}:}
%Let $g_1(\tau) = \kappa$ then $g_1(\tau) \in o\brac{e^{\frac{k_1}{\tau}}}$ and  $g_1(\tau) \in \omega\brac{e^{-\frac{k_1}{\tau}}}$, $k_1>0$. 
Since $\text{Res}(f)$ exists we have
\begin{equation}
\lim_{\tau \to 0} \frac{f(\tau)}{g(\tau) e^{-\frac{\text{Res}(f)}{\tau}}}= 1.
\end{equation}
Inversing both sides of the above equation, we get
\begin{equation}
\lim_{\tau \to 0} \frac{\frac{1}{f(\tau)}}{\frac{1}{g(\tau)}e^{-\frac{-\text{Res}(f)}{\tau}}} = 1.
\end{equation}
Therefore, we have $\text{Res}(\frac{1}{f})=-\text{Res}(f)$.

\emph{Proof of rule~\ref{eq:Res_mono}:}
Assume that $\text{Res}(f_1)< \text{Res}(f_2)$.
Using Lemma~\ref{lem:order_g_1_g2}, we have $g_2(\tau)e^{-\text{Res}(f_2)/\tau} \in o \brac{g_1(\tau)e^{-\text{Res}(f_1)/\tau}}$ and
$h_2 \in o\brac{g_1(\tau)e^{-\text{Res}(f_1)/\tau}}$.
\begin{align}
f_1& \leq f_2,\\
g_1(\tau)e^{-\text{Res}(f_1)/\tau}+h_1(\tau) &\leq g_2(\tau)e^{-\text{Res}(f_2)/\tau}+h_2(\tau) ,\\
1+\frac{h_1(\tau) }{g_1(\tau)e^{-\text{Res}(f_1)/\tau}}&\leq \frac{g_2(\tau)e^{-\text{Res}(f_2)/\tau}+h_2(\tau)}{g_1(\tau)e^{-\text{Res}(f_1)/\tau}}.
\end{align}
As $\tau$ tends to zero we arrive at a contradiction that $1\leq 0$. Therefore, $\text{Res}(f_1)\geq \text{Res}(f_2)$.

\emph{Proof of rule~\ref{eq:Res_mono_existence}:}
We have  $1 \leq \frac{f(\tau)}{f_1(\tau)} \leq \frac{f_2(\tau)}{f_1(\tau)}$ and $\text{Res}\brac{\frac{f_2(\tau)}{f_1(\tau)}} =  \text{Res}(f_1) - \text{Res}(f_2)=0$. By Rule~\ref{eq:Res_const} $ \frac{f_2(\tau)}{f_1(\tau)}$ is sub-exponential. This implies that
$ \frac{f(\tau)}{f_1(\tau)}$ is also sub-exponential
Therefore, there exists $g_{01}(\tau)$ such that

\begin{align}
1 &= \lim_{\tau \to 0} \frac{\frac{f(\tau)}{f_1(\tau)}}{g_{01}(\tau)},\\
&=\lim_{\tau \to 0} \frac{f(\tau)}{g_{01}(\tau)g_1(\tau)e^{-\frac{\text{Res}(f_1)}{\tau}}}
\lim_{\tau \to 0} \frac{g_1(\tau)e^{-\frac{\text{Res}(f_1)}{\tau}}}{f_1(\tau)} ,\\
&=\lim_{\tau \to 0} \frac{f(\tau)}{g_{01}(\tau)g_1(\tau)e^{-\frac{\text{Res}(f_1)}{\tau}}},
\end{align}
where the product $g_{01}(\tau)g_1(\tau)$ is also a sub-exponential function.
Therefore,  $\text{Res}(f)$ exists and $\text{Res}(f)=\text{Res}(f_1)=\text{Res}(f_2)$.

\end{IEEEproof}

As an illustration of application of the above rules we calculate the resistance of BLLA with deterministic utilities using the above rules.
Let $m_i(t)$ denote the probability of choosing player $i$ to revise its action.
In case of deterministic utilities, the transition probability $P^\tau_{ab}$ of BLLA is 
\begin{equation}
P^\tau_{ab} = m_i(t)\frac{e^{\frac{1}{\tau}U_i(b)}}{e^{\frac{1}{\tau}U_i(a)}+e^{\frac{1}{\tau}U_i(b)}}.
\end{equation}
Let $\Delta_i  = U_i(a)-U_i(b)$.
%, where 
%\begin{equation}
%\Delta^{+}_i  =\begin{cases}
% \Delta_i, &\quad \text{if } \Delta_i\geq 0,\\
% 0, &\quad \text{otherwise}.
%\end{cases}
%\end{equation}

Using Lemma~\ref{lem:comp_func}, we have 
$\text{Res}(P^\tau_{ab} )$
\begin{align}
&=\text{Res}\brac{m_i(t)}+\text{Res}\brac{\frac{e^{\frac{1}{\tau}U_i(b)}}{e^{\frac{1}{\tau}U_i(a)}+e^{\frac{1}{\tau}U_i(b)}}},\\
&=\text{Res}\brac{e^{\frac{1}{\tau}U_i(b)}}-\min\cbrac{\text{Res}\brac{e^{\frac{1}{\tau}U_i(a)}},\text{Res}\brac{e^{\frac{1}{\tau}U_i(b)}}},\\
\label{eqn:Rab_bllla} &=
\Delta^{+}_i,
\end{align}
wherer $\Delta^{+}_i = \max\cbrac{0,\Delta_i}$.

In the following, we show that the resistance of BLLA for the noisy potential CAP game $\mathcal{\hat{G}}^N$ with estimated utilities $\hat{\utility}^N_i$ is same as in~\eqn{eqn:Rab_bllla}.
For this, we need the following lemma.
\begin{lemma}\label{lem:diff_pN_i_p_i}
Let denote %$\Delta^N_i = U^N_i(a)-U^N_i(b)$, $\Delta_i = U_i(a)-U_i(b)$, 
$\Delta^N_i = \hat{\utility}^N_i(a)-\hat{\utility}^N_i(b),$
$\Delta_i  = U_i(a)-U_i(b),$
\begin{align}
\label{eqn:pN_i}p^N_i &= \expect{\brac{1+e^{\Delta^N_i/\tau}}^{-1}},\\
\label{eqn:p_i} p_i & = \brac{1+e^{\Delta_i/\tau}}^{-1},
\end{align}
and consider the event 
$A^\delta = \cbrac{\abs{\Delta^N_i-\Delta_i}< \delta}$.
Then
\begin{equation}
\label{eqn:diff_ki_tilde_ki_1}
\abs{p^N_i - p_i} \leq \delta\tau^{-1}p_i +2\prob{\bar{A}^\delta}.
\end{equation}

\end{lemma}
\begin{IEEEproof} 
Notice that the probability of transition of BLLA from action $a$ to $b$ in noisy potential game $\mathcal{\hat{G}}^N$ is 
$p^N_i= \text{Pr}^N\brac{a\to b}$
 %\Given\cbrac{\hat{\utility}^N_i(a),\hat{\utility}^N_i(b)}}$ 
 given in~\eqn{eqn:pN_i}
   and in deterministic potential game is 
   $p_i=\prob{a\to b}$
   % \Given\cbrac{U_i(a),U_i(b)}}$ 
   given in~\eqn{eqn:p_i}.
Using the law of total probability, we can write
\begin{equation}
p^N_i= \text{ Pr}^N\brac{a\to b \Given A^\delta}\prob{A^\delta}
+ \text{Pr}^N\brac{a\to b \Given \bar{A}^\delta}\prob{\bar{A}^\delta}.
\end{equation}
and 
\begin{multline}
\label{eqn:diff_ki_tilde_ki}
\abs{p^N_i - p_i} \leq \abs{\text{ Pr}^N\brac{a\to b \Given A^\delta}-p_i}\prob{A^\delta}\\+\abs{\text{ Pr}^N\brac{a\to b \Given \bar{A}^\delta}-p_i}\prob{\bar{A}^\delta}.
\end{multline}
It can be shown that the absolute value of the derivative of $p_i$ with respect to $\Delta_i$ is $\tau^{-1}p_i\brac{1-p_i} \leq \tau^{-1}p_i$.
 Therefore, we have 
\begin{equation}\label{eqn:lipchiltz_cond}
 \abs{\text{ Pr}^N\brac{a\to b \Given A^\delta}-p_i} \leq \delta\tau^{-1}p_i.
\end{equation}
Also, we bound $\abs{\text{ Pr}^N\brac{a\to b \Given \bar{A}^\delta}-p_i}\leq 2$. Substituting, this and~\eqn{eqn:lipchiltz_cond} in~\eqn{eqn:diff_ki_tilde_ki} we have~\eqn{eqn:diff_ki_tilde_ki_1}.
\end{IEEEproof}

%%%%%%%%%% Bounded noise
\begin{IEEEproof}[Proof for bounded noise case] 
Let denote the noise $Z_i = \hat{U}_i(a)-U_i(a) - \brac{\hat{U}_i(b)-U_i(b)}$.
Using Hoefding inequality for bounded independent random variables, we have
\begin{equation}\label{eqn:hoefding_ineq}
\prob{\bar{A}^\delta} =\prob{\frac{1}{N}\sum_{i=1}^N \abs{Z_i}> \delta} \leq 2\exp\brac{-2N \frac{\delta^2}{\ell^2}}.
 \end{equation}
Substituting~\eqn{eqn:hoefding_ineq} in Lemma~\ref{lem:diff_pN_i_p_i}, we have
\begin{equation}\label{eqn:tilde_ki_ineq2}
p_i\brac{ 1-\frac{\delta}{\tau}} - 4e^{-2N \frac{\delta^2}{l^2}}  \leq p^N_i \leq p_i\brac{1+\frac{\delta}{\tau}}+ 4e^{-2N \frac{\delta^2}{\ell^2}}.
 \end{equation}
Substituting the number of samples $N$ from~\eqn{eq:samplesBounded} and  $\delta = \brac{1-\xi}\tau$ in above, we have
  \begin{equation}\label{eqn:tilde_ki_ineq2}
\xi\brac{p_i-e^{-\frac{2}{\tau}}} \leq p^N_i \leq \brac{2-\xi}p_i + \xi e^{-\frac{2}{\tau}}.
 \end{equation}
As before, the transition probability $P^\tau_{ab}$ of BLLA is 
  \begin{equation}\label{eqn:Pab_ubn}
\xi m_i(t)\brac{p_i-e^{-\frac{2}{\tau}}} \leq P^\tau_{ab} \leq \brac{2-\xi}m_i(t)p_i + \xi m_i(t) e^{-\frac{2}{\tau}}.
 \end{equation}

In the following, we calculate the resistance of lower and upper bound of the above $ P^\tau_{ab}$ using Lemma~\ref{lem:comp_func}.
Note that $\text{Res}(p_i) = \Delta^{+}_i$, $\text{Res}(e^{-\frac{2}{\tau}} ) = 2$, and $\Delta_i \leq 2$.
The resistance of lower bound of  $ P^\tau_{ab}$ is
$\text{Res}\brac{\xi m_i(t)\brac{p_i-e^{-\frac{2}{\tau}}}}$
\begin{align}
 &= \text{Res}\brac{\xi m_i(t)}+\text{Res}\brac{\brac{p_i-e^{-\frac{2}{\tau}}}},\\
&= \min\cbrac{\text{Res}\brac{p_i},\text{Res}\brac{e^{-\frac{2}{\tau}}}} ,\\
&= \text{Res}(p_i).
\end{align}
Similarly, the resistance of upper bound of  $ P^\tau_{ab}$ is
$\text{Res} \brac{\brac{2-\xi}m_i(t)p_i + \xi m_i(t) e^{-\frac{2}{\tau}}}$
\begin{align}
 &= \min\cbrac{\text{Res}\brac{\brac{2-\xi}m_i(t)p_i},\text{Res}\brac{\xi m_i(t) e^{-\frac{2}{\tau}}}} ,\\
 &= \min\cbrac{\text{Res}\brac{p_i},\text{Res}\brac{e^{-\frac{2}{\tau}}}} ,\\
&= \text{Res}(p_i).
\end{align}
Since both the bounds have the same resistance, by Rule~\ref{eq:Res_mono_existence} the resistance of $P^\tau_{ab}$ exists and is equal to $\text{Res}(p_i)$.
Therefore, the  resistance of transitions of BLLA with bounded noise is same as in the case of without noise~\eqn{eqn:Rab_bllla}.
\end{IEEEproof}

\begin{IEEEproof}[Proof for unbounded noise case] 
In this case, we use Chernoff bound to calculate $\prob{\bar{A}^\delta}$ because of the unbounded noise as below.
Let denote the noise $Z_i$ with moment generating function $M(\theta)$.
\begin{align}
\label{eqn:prob_A_1} \prob{\bar{A}^\delta}  &=\prob{\frac{1}{N}\sum_{i=1}^N \abs{Z_i}> \delta}%,\\
=2\prob{\frac{1}{N}\sum_{i=1}^N Z_i> \delta},\\
\label{eqn:chernoff_bound}&\leq2\exp\brac{-N\log\brac{\frac{e^{\theta^*\delta}}{M(\theta^*)}}},
 \end{align}
 where,~\eqn{eqn:prob_A_1} is obtained by assuming symmetric probability distribution of noise. However, for non-symmetric distribution a more complex expression can be obtained. Also, we used the Chernoff bound for independent and identically distributed random variables to obtain the equation~\eqn{eqn:chernoff_bound}.
 
Substituting~\eqn{eqn:chernoff_bound}, $\delta = \brac{1-\xi}\tau$, and the number of samples $N$ from~\eqn{eq:samplesUnbounded} in Lemma~\ref{lem:diff_pN_i_p_i}, we have
 \begin{equation}\label{eqn:tilde_ki_ineq2}
\xi p_i-4\xi e^{-\frac{2}{\tau}} \leq p^N_i \leq \brac{2-\xi}p_i + 4\xi e^{-\frac{2}{\tau}}.
 \end{equation}

Following the same steps as before, we get that the resistance of transitions of BLLA with unbounded noise is same as in the case of without noise~\eqn{eqn:Rab_bllla}.

\end{IEEEproof}

\subsection{Proof of Convergence of BLLA with Decreasing~$\tau(t)$}\label{app:conv_blla_decreasing}

We give the proof in the case of bounded noise. The proof for unbounded noise can be done similarly. The proof is divided into several lemmas.
For a given parameter $\tau$, we fix $N(\tau)$ as in~\eqn{eq:samplesBounded}, and we consider $p(\tau) = p^{N(\tau)}$.
%For ease of notation we do not mention $i$.
Recall that $p(\tau) = \mathbb{E}\left[ f(\Delta_N, \tau) \right]$
with $ f(\delta, \tau) = \left( 1+\exp\left(\frac{\delta}{\tau}\right)
\right)^{-1}$.
\begin{lemma}\label{lem:p_deriv_prop}
  Function $\displaystyle \frac{\partial f(\delta, \tau)}{\partial \tau}$ is odd, has the sign of $\delta$, is bounded in absolute value by $k/\tau$ for some $k>0$, and the maximum is attained (for positive value) at the point $a^*\tau$, where $a^* > 0$.
\end{lemma}

\begin{IEEEproof}
  We have
  \begin{equation}
    \label{eq:derivative}
    \frac{\partial f(\delta, \tau)}{\partial \tau} = \frac{\delta}{\tau^2}\frac{1}{2+\exp(\delta / \tau)+\exp(-\delta / \tau)}.
  \end{equation}
  This is an odd function in $\delta$ that has the sign of $\delta$. Hence, we just consider the case $\delta>0$.
  Then
  \begin{equation*}
    \frac{\partial^2 f(\delta, \tau)}{\partial \delta \partial \tau} = \frac{1}{\tau^2(2+Y+Y^{-1})^2}\left[ 1 - \frac{\delta}{\tau}  \frac{ Y - Y^{-1}  }{2+ Y + Y^{-1}}  \right],
   \end{equation*}
with $Y = \exp(\delta / \tau)$. This is first positive and then negative when $\delta$ is positive. The maximum is reached when
  \begin{equation}
    \label{eq:max}
    \frac{\delta}{\tau}  \frac{ Y - Y^{-1}  }{2+ Y + Y^{-1}} = 1.
  \end{equation}
We claim that the maximum in $\delta$ is attained for $\delta^* = a^* \tau$, with $a^*>0$ a constant. Indeed, consider $\delta = a\tau$ with $a>0$ in~\eqn{eq:max}, which gives 
  \begin{equation*}
  2+\exp(a)(1-a)+\exp(-a)(1+a) = 0.
  \end{equation*}
  Consider the function $g(a) = 2+\exp(a)(1-a)+\exp(-a)(1+a)$. We have $g(0) = 4$, and $g$ tends to $-\infty$ when $a$ goes to $\infty$. Furthermore, the derivative is $-a(\exp(a)+\exp(-a))$ which is strictly negative, hence there is a unique solution $a^*$ to the equation~\eqn{eq:max}. Replacing $\delta$ by $\delta^* = a^* \tau$ in \eqn{eq:derivative} yields:
 \begin{equation*}
  \frac{\partial f(\delta^*, \tau)}{\partial \tau} = \frac{a^*}{\tau}\frac{1}{2+\exp(a^*)+\exp(-a^*)}.
  \end{equation*}
Hence the result follows with $\displaystyle k = \frac{a^*}{2+\exp(a^*)+\exp(-a^*)}$.
\end{IEEEproof}

\begin{lemma}
  \label{lem:monotony}
  If $\delta > 0$ (resp. $\delta < 0$), then $p(\tau)$ is increasing (resp. decreasing) in the vicinity of $\tau = 0$. 
  Furthermore, $\abs{p'(\tau)}$ has resistance $\delta$.
\end{lemma}

\begin{IEEEproof}
  We consider $\delta > 0$. The case $\delta <0$ is  similar.

  We will show that the derivative $p'(\tau)$ is positive in the vicinity of 0. Previous lemma shows that $\displaystyle \frac{\partial f(\delta, \tau)}{\partial \tau} \leq k/\tau$. Since the constant function $k/\tau$ is integrable w.r.t. to the distribution of $\Delta_N$, then 
 \begin{equation*}
  p'(\tau) =  \mathbb{E}\sbrac{  \frac{\partial f(\Delta_N, \tau)}{\partial \tau}  }.
 \end{equation*}
  By previous lemma, the point reaching the maximum of $\displaystyle \frac{\partial f(\delta, \tau)}{\partial \tau}$ goes to zero when $\tau$ goes to zero, and the function is then decreasing. Hence, for any $\epsilon$, there is $\tau$ small enough such that the minimum (resp. maximum) of the derivative on the interval $[\delta - \epsilon, \delta + \epsilon]$  is attained at $\delta + \epsilon$ (resp. $\delta - \epsilon$). Consider the event 
\begin{equation}
  A^\epsilon = \cbrac{\abs{\Delta^N_i-\Delta_i}< \epsilon}.
\end{equation}  
Following the  proof techniques used to show the convergence of BLLA
\begin{align}
    p'(\tau) & =  \mathbb{E}\left[  \frac{\partial f(\Delta_N, \tau)}{\partial \tau}  \right],\\
    & = \mathbb{E}\left[  \frac{\partial f(\Delta_N, \tau)}{\partial \tau} | \bar{A}^\epsilon   \right] \mathbb{P}[\bar{A}^\epsilon ] +  \mathbb{E}\left[  \frac{\partial f(\Delta_N, \tau)}{\partial \tau} | A^\epsilon \right] \mathbb{P}[ A^\epsilon ],\\
\label{eq:p1_lb_1}    & \geq -\frac{k}{\tau} \mathbb{P}[\bar{A}^\epsilon ] + \frac{\partial f(\delta+\epsilon, \tau)}{\partial \tau}\mathbb{P}[ A^\epsilon ],\\
\label{eq:p1_lb_2}     & \geq -\frac{k\xi}{\tau} \exp(-\frac{2}{\tau})+ 0.5\frac{\partial f(\delta+\epsilon, \tau)}{\partial \tau}.
  \end{align}
In the above~\eqn{eq:p1_lb_1} is obtained by using Lemma~\ref{lem:p_deriv_prop} and~\eqn{eq:p1_lb_2} is obtained by choosing $\delta = (1-\xi)\tau$. 
 Note that as in~\eqn{eq:derivative} the above second term is equivalent to $\displaystyle \frac{\delta+\epsilon}{\tau^2}\exp(-\frac{\delta+\epsilon}{\tau})$, which is a dominant term compared to $\displaystyle \frac{k\xi}{\tau} \exp(-\frac{2}{\tau})$ if $\tau$ is small enough. Hence the derivative is lower bounded by a positive function and then is positive. More, by choosing $\epsilon = \tau^2$, we see that the derivative is lower bounded by a function equivalent to $\frac{1}{\tau^2}e^{-\frac{\delta} {\tau}}$, which has the resistance $\delta$. 

  The upper bound is obtained with the following inequality:
\begin{align*}
  p'(\tau)  & = \mathbb{E}\left[  \frac{\partial f(\Delta_N, \tau)}{\partial \tau} | \bar{A}^\epsilon   \right] \mathbb{P}[\bar{A}^\epsilon ] +  \mathbb{E}\left[  \frac{\partial f(\Delta_N, \tau)}{\partial \tau} | A^\epsilon \right] \mathbb{P}[ A^\epsilon ],\\
    & \leq \frac{k}{\tau} \mathbb{P}[\bar{A}^\epsilon ] + \frac{\partial f(\delta-\epsilon, \tau)}{\partial \tau}\mathbb{P}[ A^\epsilon ],\\
    & \leq \frac{k\xi}{\tau} \exp(-\frac{2}{\tau})+ 0.5\frac{\partial f(\delta-\epsilon, \tau)}{\partial \tau}.
  \end{align*}
  And by choosing $\epsilon = \tau^2$ we obtain the same equivalent function $\frac{1}{\tau^2}e^{-\frac{\delta} {\tau}}$, which has the resistance $\delta$,
  
  Therefore, from Rule~\ref{eq:Res_mono_existence} we have that the resistance of $\abs{p'(\tau)}$ is $\delta$.
\end{IEEEproof}

\begin{lemma}
  \label{lem.weak_ergo}
  The non-homogeneous Markov chain generated by the BLLA algorithm with decreasing parameter $\tau(t) = \frac{1}{\log(1+t)}$ is weakly ergodic.
\end{lemma}

\begin{IEEEproof}
  The conditions of validity of Theorem 2 in~\cite{anily1987simulated}
  are checked by Lemma~\ref{lem:monotony}, Equation~(\ref{eqn:tilde_ki_ineq2}) and the
  classical choice of decreasing parameter $\tau$. More details about
  weak ergodicity can also be found in~\cite{bremaud2013markov}.
\end{IEEEproof}

If a real valued function $f$ is defined on the interval $\sbrac{a,b}$, $f$ is differentiable and its derivative $f'$ is Riemann integrable then its total variation $V_a^b(f)$ is
\begin{equation}
V_a^b(f) = \int_a^b \abs{f'(x)}dx.
\end{equation}
$f$ is bounded variation function if its total variation is finite i.e., $V_a^b(f)< \infty$. If the derivative $f'$ is bounded then $V_a^b(f)< \infty$ and $f$ is bounded variation function.

Let $\pi(\tau)$ be the stationary distribution of the homogeneous Markov chain  for a given $\tau$.
\begin{lemma}
  \label{lem:bound_deriv}
  $\pi(\tau)$ has a bounded derivative.
\end{lemma}

\begin{IEEEproof}
  By the Markov chain tree theorem~\cite{anantharam1989proof} for every state $c\in X$,  we have  
  $\pi_c(\tau) = \frac{u(c)}{\sum_{d\in X}u(d)}$ where
\begin{equation}\label{eq:u_c}
  u_c(\tau) = \sum_{T \in \mathcal{T}_c} \prod_{e \in T} p_e(\tau),
\end{equation}
$p_e(\tau)$~\eqn{eqn:pN_i} is transition probability to state $e$ 
  and $\mathcal{T}_c$ is the set of trees rooted in state $c$.
  Then
  \begin{align*}
    \abs{\pi_c'(\tau)} & = \left| \left( \frac{u_c}{\sum_{d}u_{d}} \right)' \right|,  \\
    &  \leq    \frac{|u'_c|(\sum_{d}u_{d})}{(\sum_{d}u_{d})^2} +   \frac{u_c\abs{\sum_{d}u'_{d} }}{(\sum_{d}u_{d})^2},\\
    &  \leq    \frac{|u_c'|}{\sum_{d}u_{d}} +   \frac{| \sum_{d}u'_{d} |}{\sum_{d}u_{d}},\\
    &  \leq    \frac{|u_c'|}{\sum_{d}u_{d}} +   \sum_{d}\frac{| u'_{d} |}{\sum_{d'}u_{d'}}.
  \end{align*}
  Hence it suffices to show that $\displaystyle \frac{|u_c'|}{\sum_{d}u_{d}}$ is bounded for all states $c$. 
  
% Elementary computation below shows that $\sum_{d}u_{d}$ is equivalent to $\exp(-R_{\min} / \tau)$ where $R_{\min}$ is the minimal resistance of a tree.
 % (the resistance of a transition is the same as in the deterministic case, as is shown in the convergence proof of BLLA).
 Let $ \text{Res}\brac{T}$ denotes the total resistance of a tree $T$ and $R_{\min}$ denotes the resistance of the minimal resistance tree.
By using Lemma~\ref{lem:comp_func}, we obtain 
 \begin{align}
 \text{Res}\brac{\sum_{d}u_d(\tau)} &= \text{Res}\brac{\sum_{d} \sum_{T \in \mathcal{T}_d} \prod_{e \in T} p_e(\tau)},\\
=& \min_{d\in X}\min_{T \in \mathcal{T}_d} \text{Res}\brac{\prod_{e \in T} p_e(\tau)},\\
  =& \min_{d\in X, T \in \mathcal{T}_d} \text{Res}\brac{T},\\
 = &R_{\min}.
 \end{align}
%Therefore, by Lemma~\ref{lem:comp_func} we have $\sum_{c}u_c(\tau) \sim e^{-\frac{R_{\min}}{\tau}}$.

The derivative of transition probability $u'_c(\tau)$ is obtained using~\eqn{eq:u_c} as
\begin{equation}
u'_c(\tau) = \sum_{T \in \mathcal{T}_c}\sum_{e \in T} p'_e(\tau)\prod_{d \in T / e} p_d(\tau),
\end{equation}
where $p'_e(\tau)$ is equivalent to $\exp(-|\delta|/\tau)/\tau^2$ by Lemma~\ref{lem:monotony}.
The resistance of $u'_c(\tau)$ is 
 \begin{equation}
\text{Res}\brac{u'_c} %&= \sum_{T \in \mathcal{T}_c} p'_e(\tau)\prod_{d \in T / e} p_d(\tau),\\
 = \min_{T \in \mathcal{T}_c}\min_{e\in T}\sbrac{ \text{Res}\brac{p'_e}+ \sum_{d \in T/e}\text{Res}\brac{ p_d(\tau)}}. 
 \end{equation}
 Since by Lemma~\ref{lem:comp_func} ,  $\text{Res}\brac{p_e'}$  is $\abs{\delta}$ that is same as that of $\text{Res}\brac{p_e}$~\eqn{eqn:Rab_bllla}
  if $\delta > 0$ and is strictly greater if $\delta < 0$. Therefore, we have $\text{Res}\brac{p_e'} \geq\text{Res}\brac{p_e}$ 
  and  $\text{Res}\brac{u'_c}\geq \text{Res}\brac{\sum_{c}u_c(\tau)}$.
 But the minimal resistance tree must contain a transition with null resistance (which corresponds to the best response). 
 
 \begin{lemma}\label{lem:min_res_tree}
 A minimum resistance tree must contain a transition with zero resistance.
 \end{lemma}
 \begin{IEEEproof}
 Assume that a minimum resistance tree $T_{\min}$ have all the transitions with non-zero resistance. 
 Let the root of this tree be a state $s$ and let there be a transition from another state $s'$ to $s$. 
 Let $R_{s'\to s}$ be a non-zero resistance of this transition. 
 Note that the resistance of reverse transition $R_{s\to s'}=0$ because it corresponds to the best response transition.
 Construct a new tree $T$ rooted at state $s'$ by adding the transition $s \to s'$ and removing the transition $s' \to s$. The resistance of the tree $T$ is 
\begin{align}
R_T &= R_{T_{\min}}-R_{s'\to s}+R_{s\to s'},\\
&<R_{T_{\min}}.
\end{align}
We arrive at a contradiction. Therefore, a minimum resistance tree must contain a transition with null resistance.
  \end{IEEEproof}

 Hence, the state $c$ at which $R_{\min}$ is reached  contains at least a transition with $\delta \leq 0$.
 Therefore,  $\text{Res}\brac{u'_c}>R_{\min}$. 
Using Lemma~\ref{lem:comp_func}, we have 
\begin{equation}
\frac{|u_c'|}{\sum_{d}u_{d}} = \kappa\exp\brac{-\frac{\text{Res}\brac{u'_c}-R_{\min}}{\tau}}+h_1(\tau),
\end{equation}
where $h_1(\tau) \in o\brac{\exp\brac{-\frac{\text{Res}\brac{u'_c}-R_{\min}}{\tau}}}$. Observe from the above equation that 
$\frac{|u_c'|}{\sum_{d}u_{d}}\to 0$ as 
$\tau$ goes to zero for all states $c$. This finally shows that the derivative $\abs{\pi_c'(\tau)}$ is bounded.
\end{IEEEproof}

\begin{IEEEproof}[Proof of Theorem~\ref{thm:decTau}] 
  We check the assumptions of Theorem~1 in~\cite{anily1987ergodicity} are satisfied for the proof of Theorem~\ref{thm:decTau}. 
  By Lemma~\ref{lem.weak_ergo}, the algorithm generates a weakly ergodic non-homogeneous Markov chain.
  Lemma~\ref{lem:bound_deriv} shows that the stationary distribution $\pi(\tau)$  of the homogeneous Markov chain is a bounded variation function of $\tau$ (this is a direct consequence of derivative of $\pi(\tau)$  being bounded). 

\end{IEEEproof}

\end{document}